\documentclass{ecai} 

\usepackage{latexsym}
\usepackage{amssymb}
\usepackage{amsmath}
\usepackage{amsthm}
\usepackage{booktabs}
\usepackage{enumitem}
\usepackage{graphicx}
\usepackage{color}
\usepackage{hyperref}
\usepackage[linesnumbered, lined, boxed, commentsnumbered, ruled]{algorithm2e}

\newtheorem{theorem}{Theorem}

\newtheorem{corollary}[theorem]{Corollary}
\newtheorem{proposition}[theorem]{Proposition}

\newtheorem{definition}{Definition}

\newcommand{\BibTeX}{B\kern-.05em{\sc i\kern-.025em b}\kern-.08em\TeX}

\usepackage[table, xcdraw]{xcolor}

\usepackage{amsthm}
\usepackage{amsmath}
\usepackage{amssymb}
\newtheorem{remark}{Remark}[]
\newtheorem{example}[]{Example}

\usepackage{hhline, colortbl}
\usepackage{subcaption}
\usepackage{hhline} 

\newcommand{\Imp}[0]{\mathit{IMP}}

\begin{document}


\begin{frontmatter}


\paperid{1827} 


\title{
Limited Voting for Better Representation?
}


\author[A]{\fnms{Maaike}~\snm{Venema-Los}\orcid{0000-0002-0704-2081}}
\author[A]{\fnms{Zo\'{e}}~\snm{Christoff}\orcid{0000-0003-2412-8458}}
\author[A,B]{\fnms{Davide}~\snm{Grossi}\orcid{0000-0002-9709-030X}\thanks{Corresponding Author. Email: d.grossi@rug.nl}} 

\address[A]{
University of Groningen
}
\address[B]{
University of Amsterdam
}


\begin{abstract}
Limited Voting (LV) is an approval-based method for multi-winner elections where all ballots are required to have a same fixed size. 
While it appears to be used as voting method in corporate governance and has some political applications, to the best of our knowledge, no formal analysis of the rule exists to date. We provide such an analysis here, prompted by a request for advice about this voting rule by a health insurance company in the Netherlands, which uses it to elect its work council.
We study conditions under which LV would improve representation over standard approval voting and when it would not. We establish the extent of such an improvement, or lack thereof, both in terms of diversity and proportionality notions. These results help us understand if, and how, LV may be used as a low-effort fix of approval voting in order to enhance representation.  
\end{abstract}

\end{frontmatter}



\section{Introduction}
    Limited Voting (LV) is a form of multi-winner approval voting where ballots are limited: 
    voters can only submit a ballot of exactly $l$ votes (where $l$ is at most the committee size $k$); the $k$ candidates with largest number of votes are selected. Although the rule is used in corporate governance and some political contexts,\footnote{See, for instance, \url{https://en.wikipedia.org/wiki/Limited_voting} for an overview of some political use cases.} it has not yet been analysed formally, as far as we know. 
    We have been approached by 
    a health insurance 
    company from the Netherlands, that uses LV to elect its work council and wonders about its pros and cons.
    In this paper we report on an initial analysis we made of LV in order to gain a better understanding of the behavior of the rule. Since the main goal of a work council is to represent the employees of the company, we are concerned with representation in LV. Our working hypothesis is that LV is being deployed as an easy adjustment to the naive baseline of standard approval voting (AV), with the intuition that it can limit its over-representation effects, at least in settings where voters are clustered in different groups with similar opinions. 
   The aim of the paper is not to argue in favor of LV but rather to rationalize the above intuition and understand its exact scope.
    
\paragraph{Motivation}
At first sight, 
LV does not seem very appealing. 
Indeed, from the point of view of the voters, the rule is literally limiting them in what they are allowed to express
. Some voters may have less than $l$ approved candidates, and LV forces them to vote for candidates they do not approve. Other
s may 
approve more than $l$ 
candidates, and cannot vote for all candidates they approve. 
This loss of information about voters' 
preferences, compared to AV,
makes 
the rule inefficient. It does not, for example, satisfy Pareto efficiency
(see Example \ref{ex:pareto} in Appendix \ref{app:additional}). 
We can establish the worst possible welfare loss of using LV instead of AV as follows.  Let us assume that the voters' welfare increases linearly with the number of approved candidates they have in the winning committee: for each candidate that they approve who is elected, they get one `unit of welfare'. We can measure the welfare performance of a rule by looking at the sum of all voters' welfare, given the elected committee. With this measure---called the `AV-score' of a committee---it is easy to construct examples where LV performs arbitrarily worse than AV. 

\begin{example}\label{ex:bad_LV}
   In the election shown in Table \ref{tab:AV-guarantee},  the committee $W=\{c_1, c_2, c_3, c_4\}$ is a winning committee by LV (all its members get one vote and no candidate in the total election gets more than one vote), while any committee $W'\subseteq \{c_5, c_6, c_7, c_8, c_9\}$ (with $|W'| = k = 4$) wins by AV. The AV-score of $W$ is 4, since every elected candidate is only approved by one voter, while the AV-score of $W'$ is 24, since every elected candidate is approved by all six voters. 
\end{example}

\begin{table}[b]
    \small
    \centering
    \renewcommand{\arraystretch}{0.7}
    \setlength{\tabcolsep}{3pt}
        \begin{tabular}{c|c|c|c|c|c|c|c|c|c|}
                    \multicolumn{1}{c}{}& \multicolumn{1}{c}{
                    $c_1$} & \multicolumn{1}{c}{$c_2$} & \multicolumn{1}{c}{$c_3$} & \multicolumn{1}{c}{$c_4$} & \multicolumn{1}{c}{
                    $c_5$} & \multicolumn{1}{c}{$c_6$} & \multicolumn{1}{c}{$c_7$} & \multicolumn{1}{c}{$c_8$} & \multicolumn{1}{c}{$c_9$}  \\
                    \hhline{~---------}
            $v_1$ & \cellcolor[HTML]{C1C1C1} x &  &  &    &\cellcolor[HTML]{C1C1C1} &\cellcolor[HTML]{C1C1C1}  & \cellcolor[HTML]{C1C1C1}  &\cellcolor[HTML]{C1C1C1} &\cellcolor[HTML]{C1C1C1}\\ 
            \hhline{~---------}
            $v_2$ & & \cellcolor[HTML]{C1C1C1} x   &    &  &\cellcolor[HTML]{C1C1C1} &\cellcolor[HTML]{C1C1C1}  & \cellcolor[HTML]{C1C1C1}  &\cellcolor[HTML]{C1C1C1} &\cellcolor[HTML]{C1C1C1}\\ 
            \cline{2-10} 
            \hhline{~---------}
            $v_3$ & &  & \cellcolor[HTML]{C1C1C1} x    &  &\cellcolor[HTML]{C1C1C1} &\cellcolor[HTML]{C1C1C1}  & \cellcolor[HTML]{C1C1C1}  &\cellcolor[HTML]{C1C1C1} &\cellcolor[HTML]{C1C1C1}\\ 
            \hhline{~---------}
            $v_4$ & &  &  &  \cellcolor[HTML]{C1C1C1} x   &\cellcolor[HTML]{C1C1C1} &\cellcolor[HTML]{C1C1C1}  & \cellcolor[HTML]{C1C1C1}  &\cellcolor[HTML]{C1C1C1} &\cellcolor[HTML]{C1C1C1}\\ 
            \hhline{~---------}
           $v_5$ & &  &  &  & \cellcolor[HTML]{C1C1C1} x  &\cellcolor[HTML]{C1C1C1} &\cellcolor[HTML]{C1C1C1}  & \cellcolor[HTML]{C1C1C1}  &\cellcolor[HTML]{C1C1C1} \\ 
           \hhline{~---------}
           $v_6$ & &  &  &  & \cellcolor[HTML]{C1C1C1}   &\cellcolor[HTML]{C1C1C1} x &\cellcolor[HTML]{C1C1C1}  & \cellcolor[HTML]{C1C1C1}  &\cellcolor[HTML]{C1C1C1} \\ \hline \hline
           LV & \cellcolor[HTML]{1d53ab}&  \cellcolor[HTML]{1d53ab}&  \cellcolor[HTML]{1d53ab}&  \cellcolor[HTML]{1d53ab}&   & &  &  & \\ 
           \hhline{~---------}
           AV &    & &  &  & \cellcolor[HTML]{d64238}&  \cellcolor[HTML]{d64238}&  \cellcolor[HTML]{d64238}&  \cellcolor[HTML]{d64238} & \\ 
           \hhline{~---------}
        \end{tabular}
        \vspace{0.2cm}
           \caption{Rows indicate voters $v_i$, columns candidates $c_i$, $l=1$, $k=4$. Approvals in grey, votes indicated by x. In the two bottom rows, a winning committee according to LV (first row, blue) and a winning committee according to AV (second row, red) are indicated.
   }
   \label{tab:AV-guarantee}
    \end{table}

The example generalizes to obtain arbitrarily large differences in AV-score, where any voter $v_i$ approves $A_i = \{c_i\cup\{c_j: k+1\leq j \leq m\}\}$ (with $m$ the number of candidates) and votes for $c_i$ only. Then, $W=\{c_1, ..., c_k\}$ is a winning committee by LV: no candidate gets more than one vote and all candidates in $W$ get one vote. Yet, its AV-score is only $k$, while any committee $W'\subseteq \{c_j: k+1\leq j \leq m\}$ (with $|W'|=k$) has an AV-score of $k\cdot n$ (with $n$ the number of voters). The worst case proportion between a rule's AV-score in an election and the maximum AV-score possible in that election is a rule's \emph{AV-guarantee} \citep{LacknerSkowron2019}. The above construction shows that the AV-guarantee of LV tends to $0$ as $n$ grows. 

\medskip
 
So, why is LV deployed in practice? The motivation may lie in the fact that LV can restrict the influence of large pluralities on voting outcomes, without changing the way in which outcomes are computed. Example \ref{ex:introduction_example} below illustrates this intuition: under some assumptions on the voters' preferences and ballots, LV gives representation to voters that would otherwise not be represented under AV.
    
    \begin{example}\label{ex:introduction_example}
        Suppose we have three sets of candidates (parties) $C_a = \{a_1, a_2, ...\}$, $C_b = \{b_1, b_2, ...\}$, and $C_c = \{c_1, c_2, ...\}$ and three groups of voters $V_a, V_b, V_c$ 
        who approve the respective parties, and suppose the desired committee size $k=8$.  
        Then, if $|V_a|> |V_b|>|V_c|$ and $|C_a| \geq 8$, with AV, the winning committee will be completely filled with candidates from $C_a$. With LV and limit $l=3$, if all voters vote for the first (by index) $3$ candidates of their party, any winning committee consists of $8$ candidates from $\{a_1, a_2, a_3, b_1, b_2, b_3, c_1, c_2, c_3\}$, thereby providing representation also to $V_b$ and $V_c$. 
    \end{example}

\paragraph{Contribution}
    In this paper, we study the extent to which the intuition illustrated above can rationalize the use of LV in specific voting settings.
    As a measure of diversity, we use the Chamberlin-Courant (CC) score \cite{thiele1895,chamberlincourant1983},
    and we compare LV's diversity to that of AV to get LV's \emph{CC-improvement} in a given election (Section \ref{sec:div}). 
    We show that in elections where parties are able to coordinate their voters about which candidates to include in their ballots---to which we refer as {\em broadcasted party-list elections}---LV can indeed provide higher diversity than AV. 
    We show that this diversity increase can be measured as a function of the number of parties, the committee size and the ballot size. 
    We then consider LV's proportionality (Section \ref{sec:prop}) and show that it 
    fails to satisfy all common proportionality axioms. 
    We use a similar quantitative measure, the 
    \emph{PAV-improvement}, to compare LV's proportionality to that of AV, and show that in broadcasted party-list elections, if the size of the parties is not too dissimilar, 
    LV  is more proportional than AV.
    We then show with computational experiments (Section \ref{sec:experiment}) that the more an election deviates from a broadcasted party-list election, the smaller the benefit of LV over AV. Finally (Section \ref{sec:games}), we study how voters within parties could strategically coordinate on whom to vote for, in order to maximize the number of their representatives. Interestingly, we observe that the committees selected at equilibria of this type of games do provide good proportionality guarantees. 
    
    We present the main proofs in the body of the paper. Remaining proofs as well as additional results and examples are provided as supplementary material in the appendices. 


\section{Related work}

Although LV has not yet been the object of direct analysis, 
similar multi-winner voting rules are discussed in the literature. One is Bloc/Block Voting: in \citet[chapter 2]{endriss_2017} and \citet{Elkind2017}, `Bloc (Voting)' refers to the case where the ballot size $l$ is exactly equal to the committee size $k$. (In \citet{Kamwa_Merlin_2015}, this is called `Limited Voting'.) In \citet{Bredereck2021}, $\ell$-Bloc is used: every voter \textit{approves} (and votes for) her $\ell$ most favorite candidates, assuming ordinal preferences. 
Similarly, but for single-winner, \citet{moulin2016} mentions $k$-approval voting, where a voter votes for the first $k$ candidates in their ranking. In \citet{Janson2016}, `Block Vote' allows voters to submit a ballot of length \emph{at most} $k$ (so the ballots of different voters can have different lengths), instead of insisting on length $l$. In \citet{Baharad2005} another voting rule similar to LV, Restricted Approval Voting (RAV), is introduced, in the setting where voters have \emph{ordinal} preferences, in single-winner elections. RAV is similar to LV in the idea of setting a limit on the ballot sizes, but it  gives more freedom to voters: they can approve between a lower bound $l$ and an upper bound $u$ of candidates. We expand on this specific model and its relation to ours in Remark \ref{remark:ranked}.

In \citet{Lackner2023}, Block Voting, Limited Voting, and Single Non-Transferable Vote (SNTV) are mentioned, where Limited Voting refers to the case where voters may approve \emph{at most} $l$ candidates for some $l<k$, and SNTV is Limited Voting for $l=1$. We study here the `strict' case of limited voting where voters are requested to provide \emph{exactly} $l$ approved candidates.


There is also a version of the Chamberlin Courant rule with which LV bears resemblance. This is the Thiele rule known as $l$-CC \cite{BrillLaslierSkowron2016}. Formally, $l$-CC is the approval based scoring rule with scoring vector $(1, ..., 1, 0, ..., 0)$, where the first $l$ entries are $1$ and the rest $0$. On first sight this rule is very similar to LV, since both rules are approval based rules where only $l$ of a voter's approvals are counted. The difference lays in the moment when the limit is applied. LV applies an \emph{ex ante} limit, where a voter has to specify on which $l$ of her approvals she votes before the rule is applied. $l$-CC has an \emph{ex post} limit, where information about \emph{all} approvals is the input for the voting rule, and the mechanism of the rule decides what is the optimal subset of candidates for every voter to consider. 


With respect to the reasons for limiting ballot sizes, \citet{Lee2017} observes that ballot length restrictions affect different voters to different extents, and therefore may be hard to justify 
to the voters. However, when voters have non-dichotomous preferences, restricting ballots to the size of the committee can be justified in light of strategy-proofness. In contrast, in this paper, we do assume dichotomous preferences, and we explore the extent to which ballot size restrictions may still be justified, but now on representation grounds.


\section{Preliminaries}
An {\em election} $E = ( N, C, k, l, A, L)$ is a tuple consisting of a set of voters $N = \{v_1, ..., v_n\}$ ($n \in \mathbb{N}$) and a set of candidates $C = \{c_1, ..., c_m\}$ ($m \in \mathbb{N}$), a committee size $k\leq m$, a ballot limit $1\leq l\leq k$, an approval preference profile 
$A = (A_i)_{v_i \in N}$
where $A_i\subseteq C$ is
the set of candidates that $v_i$ approves, 
and a ballot profile  
$L = (L_i)_{v_i \in N}$
where $L_i\subseteq C$ with $|L_i|=l$ is 
the set of candidates that $v_i$ votes for.
We call an election without the ballot profile $L$ an {\em election frame} (denoted by $F$).
If it is not clear from the context, we sometimes use the notation $A_E$, $k_E$, etc. to refer to the $A$, $k$, etc. of election $E$.
While in approval-based committee rules (ABC-rules, \citet{Lackner2023}), $A=L$, in our setting $A$ and $L$ may be different. Indeed, while we assume that each voter has dichotomous preferences (their approval set), their ballot is in general not equal to their approval set. 
Still, we assume that the ballot of a voter $v_i$ 
remains 
consistent with their approval set, in the following sense: if $l \leq |A_i|$, then $L_i\subseteq A_i$, and if $l>|A_i|$, $A_i\subset L_i$.  

Given such an election, 
LV
selects $k$ candidates with the highest number of votes. 
Let the LV-score of a candidate $c$ be  $s_{LV}(c)=|\{v_i\in N : c\in L_i\}|$. 
Then: 
\begin{definition}[Limited Voting (LV)]
Given election $E = ( N, C, k, l, A, L)$,  
LV elects 
a committee $W$ with size $k$ that has the highest LV-score $s_{LV}(W)=\sum_{c\in W}s_{LV}(c)$.
\end{definition}
We compare LV with 
AV, which elects 
$k$ candidates that are approved by the highest amount of 
voters. That is, the AV-score of a candidate $c$ is $s_{AV}(c)=|\{v_i\in N : c\in A_i\}|$, and AV elects 
a committee $W$ of size $k$ with the highest AV-score $s_{AV}(W) = \sum_{c\in W}s_{AV}(c)
$.
We denote the set of outcomes of LV, respectively AV (
neither of them is resolute), on an election $E$ by $LV(E)$, respectively $AV(E)$.

In the context of corporate elections, it is not unusual for candidates to be clustered in 
parties (for instance, belonging to a same office or department), and for voters to therefore support candidates belonging to a same party. In social choice theory, approval profiles where this is the case are called \emph{party-list profiles} \cite{Peters2020limwelf, Peters2021, Botan2021}
. Formally, a profile $A=\{A_1, \cdots, A_n\}$ is a party-list profile if for all $v_i, v_j\in N$, either $A_i=A_j$ or $A_i\cap A_j=\emptyset$. A party is therefore a subset of $C$ that is approved by all and only the voters in a subset of $N$.
For notational reasons, we order parties from most to least popular as $P_1, P_2, ...$ where 
$P_i$ is the set of candidates of the party (which we also use as the party's name), and
$n_i$ is the number of voters 
approving party $P_i$, so $n_1\geq n_2 \geq ...$~. We call elections (or frames) with party-list profiles, party-list elections (or frames).


We extend the notion of party-list profile by assuming that parties are able to coordinate, or signal to, their base which candidates the party wishes to see elected. We call \emph{broadcasting order} the priority by which a party's candidates are voted for by its voters. For now, we assume that parties are able to signal to their base only one such fixed order (hence the qualification `broadcasted'). Such an assumption is plausible in situations in which the party is not able to instruct voters individually. We will be lifting this assumption later in Section \ref{sec:games}.

\begin{definition}[Broadcasting order]\label{def:broadcasting_order}
    Given an election $E = ( N, C, k, l, A, L)$, a \emph{broadcasting order} $\succ_C$ is a linear order over all candidates in $C$
    . 
    We say that $E$ is consistent with $\succ_C$ iff for all voters $v_i\in N$, for any $c, c' \in A_i$: if $c\in L_i$ and $c'\notin L_i$, $c\succ_C c'$.
\end{definition} 
We say that elections are \emph{broadcasted party-list} if they are party-list and ballots are consistent with a broadcasting order. 
\begin{definition}[Broadcasted party-list elections] \label{def:bpl}
        An election $E = ( N, C, k, l, A, L)$ is a \emph{broadcasted party-list election} if $A$ is a party-list profile and 
        $E$ is consistent with a broadcasting order $\succ_C$.
\end{definition}
It follows that, in a broadcasted party-list election, for all $v_i, v_j\in N$, if $A_i=A_j$, then $L_i = L_j$. It is worth noting also that broadcasted party-list elections relate directly to elections in apportionment problems, where voters vote for parties, but parties determine which candidates receive seats \cite{balinski1994apportionment}. 

\begin{remark} \label{remark:ranked}
    In the above definitions we assumed voters to hold approval preferences and determine their limited ballots based on them.
    The model by \citet{Baharad2005} offers an alternative: voters have ordinal preferences and hold individual cutoff points $t^*_i$, which define their approval sets. Truthful ballots only contain candidates that are ranked higher than $t^*_i$. In this model, one can define party-list elections in a way that makes our broadcasting order redundant by requiring that voters: either have identical $t^*_i$s and rank candidates in the same way above this cutoff (within parties), or the set of candidates occurring above the cutoffs are disjoint (across parties).  
\end{remark}

\section{
Diversity} \label{sec:div}
The motivating example in the introduction 
shows that, 
in some situations, LV can improve the diversity of a committee compared to AV by imposing a limit on the size of the 
ballot. In this section we aim at quantifying exactly such improvements. We do so using two measures (CC-improvement and CC-guarantee) based on the gold standard of diversity in approval-based elections: the Approval-Based Chamberlin-Courant ($\alpha$-CC) rule \citep{
endriss_2017}. The rule selects the committees which maximize the Chamberlin-Courant score or {\em CC-score}: the CC-score of a committee $W$ given an approval profile $A$ is $s_{CC}(A,W)= |\{v_i\in N: W\cap A_i \not = \emptyset\}|$. In other words, the rule minimizes the number of voters left without representatives.

\subsection{CC-improvement}
To compare the diversity of the outcome of LV to that of 
the outcome of 
AV, we define the 
{\em (relative) improvement}
of LV with respect to AV, by comparing the worst scoring committee(s) selected by LV to the best scoring committee(s) selected by AV, as follows: 
\begin{definition}
[CC-improvement]
    For an election $E=( N, C, k, l, A, L)$, the 
    improvement in CC-score of LV w.r.t. AV is
    $
    \Imp_{CC}(E) = \frac{\min_{W_{LV}\in LV(E)}s_{CC}(A, W_{LV})}{\max_{W_{AV}\in AV(E)}s_{CC}(A, W_{AV}) }. 
    $
\end{definition}
We take the \emph{relative} rather than the \emph{absolute} difference between the scores of LV and AV, because the values of the scores can vary considerably between elections with different numbers of voters and candidates, which is not the difference we want to focus on.

Our starting intuition was that LV might select more diverse committees when voters can be roughly divided into 
cohesive groups. We hence start our analysis from party-list profiles. Note that in any party-list election, AV will first select all candidates from the most popular party, then from the second-most, etc., until the committee is filled. In contrast, with LV we cannot directly say what the outcome will be, since it 
depends on which of their approved candidates 
voters vote on. 
In \emph{broadcasted} party-list elections, however, we know that voters will vote according to a specified order: LV will select $l$ candidates from the most popular party, then $l$ from the second-most popular, etc., until the committee is filled.

\begin{theorem}\label{thm:CC-improvement_BP}
Let $E = ( N, C, k, l, A, L)$ be any broadcasted party-list election with parties $P_1, \ldots, P_g$. 
The CC-improvement 
is 
$IMP_{CC}(E) = \frac{\sum_{i=1}^{{\min(\lceil\frac{k}{l}\rceil, g)}}n_i}{\sum_{j=1}^{s}n_j}$
, where $s$ is the largest integer $t$ such that $\sum_{i=1}^{t}|P_i|\leq k$.
\end{theorem}

\begin{proof}
AV elects $\min(|P_1|, k)$ candidates from $P_1$, then, if less than $k$ are elected, $\min(|P_2|, k-|P_1|)$ candidates from $P_2$, etc., until $k$ candidates are elected.
Therefore, if $s$ is the number of parties from which AV elects candidates, for any $W_{AV}\in AV(E)$, the CC-score will be $s_{CC}(A,W_{AV})=\sum_{j=1}^{s}n_j$, so $\max_{W\in AV(E)}s_{CC}(A, W)=\sum_{j=1}^{s}n_j$.
LV selects $l$ members from $P_1$, $l$ members from $P_2$, etc., up until the $\lceil\frac{k}{l}\rceil$th party (if $\lceil\frac{k}{l}\rceil\leq g$, otherwise it 
takes arbitrary extra candidates that are not voted for), so for any $W_{LV}\in LV(E)$, 
$s_{CC}(A,W_{LV}) = \sum_{i=1}^{\min(\lceil\frac{k}{l}\rceil, g)}n_i$. 
\end{proof}

\begin{corollary}
In broadcasted party-list elections the CC-improvement is strictly above 1 if $\min(\lceil\frac{k}{l}\rceil, g) > s$.
\end{corollary}
From the theorem, it is also clear that the CC-score is maximised by maximising $\min(\lceil\frac{k}{l}\rceil, g)$ or, equivalently, by minimising $l$:
\begin{corollary}\label{cor:l=1}
In broadcasted party-list elections, the CC-score is maximal if $l=1$.
\end{corollary}

While Theorem \ref{thm:CC-improvement_BP} identifies precise conditions for LV to actually improve diversity, it is important to qualify the scope of the result. First, it should be noted that the assumption of a broadcasting order in Theorem \ref{thm:CC-improvement_BP} is necessary. If in a party-list profile voters do not coordinate within their party on which candidates to vote for, LV does not necessarily give a higher CC-score than AV. It can be the case that the voters from the 
largest party spread their votes over more than $l$ candidates 
to get more than $l$ representatives, or that 
they 
spread their votes too much and do not get any 
representative.
Second, even though LV may be more diverse than AV in broadcasted party-list elections, it should be stressed that LV does not maximize the CC-score: its outcome is in general not the same as the outcome of $\alpha$-CC. This rule will in such elections return a committee with at least one member of every party, if there are enough seats, and fill up the rest of the committee arbitrarily. Hence, 
if there are $g$ parties, 
$s_{CC}(\alpha\text{-CC}) = \sum_{i=1}^{\min(k, g)}n_i$. 
The 
difference between the CC-score of an outcome of LV and the maximal CC-score is therefore 
$s_{CC}(\alpha\text{-CC}) - s_{CC}(A,W_{LV}) = \sum_{i=\lceil\frac{k}{l}\rceil+1}^{\min(k, g)}n_i$.


We conclude by observing that, in general elections, LV does not necessarily improve diversity with respect to AV:

\begin{example}
   Given the profile in Table \ref{tab:CC}, the outcome of AV is $W_{AV} = \{c_2, c_3, c_4, c_5\}$, 
   representing all voters,
   while the outcome of LV is $W_{LV}=\{c_1, c_2, c_3, c_4\}$
   , leaving voter $v_6$ unrepresented. This gives a 
   CC-improvement
   of 
   $\frac{5}{6}<1$.

\begin{table}[b]
    \small
    \centering
    \renewcommand{\arraystretch}{0.65}
    \setlength{\tabcolsep}{4pt}
    \begin{tabular}{c|c|c|c|c|c|c|c|c|| c|c|}
     \multicolumn{1}{c}{}& \multicolumn{1}{c}{}& \multicolumn{1}{c}{}& \multicolumn{1}{c}{}& \multicolumn{1}{c}{}& \multicolumn{1}{c}{}& \multicolumn{1}{c}{}& \multicolumn{1}{c}{}& \multicolumn{1}{c}{}&\multicolumn{2}{c}{repr.} \\
                \multicolumn{1}{c}{}& \multicolumn{1}{c}{
                $c_1$} & \multicolumn{1}{c}{
                $c_2$} & \multicolumn{1}{c}{$c_3$} & \multicolumn{1}{c}{$c_4$} & \multicolumn{1}{c}{$c_5$} & \multicolumn{1}{c}{$c_6$} & \multicolumn{1}{c}{$c_7$} & \multicolumn{1}{c}{$c_8$} &\multicolumn{1}{c}{LV} &\multicolumn{1}{c}{AV}   \\
                \hhline{~*{10}{-}}
        $v_1$ & \cellcolor[HTML]{C1C1C1} x & \cellcolor[HTML]{C1C1C1} x & \cellcolor[HTML]{C1C1C1} x & \cellcolor[HTML]{C1C1C1} & \cellcolor[HTML]{C1C1C1} &  & & & \cellcolor[HTML]{1d53ab}& \cellcolor[HTML]{d64238}\\ 
        \hhline{~*{10}{-}}
        $v_2$ & & \cellcolor[HTML]{C1C1C1} x & \cellcolor[HTML]{C1C1C1} x & \cellcolor[HTML]{C1C1C1} x & \cellcolor[HTML]{C1C1C1} & & & & \cellcolor[HTML]{1d53ab}& \cellcolor[HTML]{d64238}\\ 
        \hhline{~*{10}{-}}
        $v_3$ & &  & \cellcolor[HTML]{C1C1C1} x & \cellcolor[HTML]{C1C1C1} x & \cellcolor[HTML]{C1C1C1} x &\cellcolor[HTML]{C1C1C1} & & & \cellcolor[HTML]{1d53ab}& \cellcolor[HTML]{d64238}\\ 
        \hhline{~*{10}{-}}
        $v_4$ & \cellcolor[HTML]{C1C1C1} x & \cellcolor[HTML]{C1C1C1} x & \cellcolor[HTML]{C1C1C1} x & \cellcolor[HTML]{C1C1C1} &  & & & & \cellcolor[HTML]{1d53ab}& \cellcolor[HTML]{d64238}\\ 
        \hhline{~*{10}{-}}
        $v_5$ & \cellcolor[HTML]{C1C1C1} x & \cellcolor[HTML]{C1C1C1} x & \cellcolor[HTML]{C1C1C1}  & \cellcolor[HTML]{C1C1C1} x &  & & & & \cellcolor[HTML]{1d53ab}& \cellcolor[HTML]{d64238}\\ 
        \hhline{~*{10}{-}}
        $v_6$ & &  & &  & \cellcolor[HTML]{C1C1C1} &\cellcolor[HTML]{C1C1C1} x & \cellcolor[HTML]{C1C1C1} x  & \cellcolor[HTML]{C1C1C1} x & & \cellcolor[HTML]{d64238}\\  \hline 
                \multicolumn{1}{c}{} & \multicolumn{8}{c}{winning committees}& \multicolumn{2}{c}{CC-score} \\ \hline
        LV & \cellcolor[HTML]{1d53ab} & \cellcolor[HTML]{1d53ab} & \cellcolor[HTML]{1d53ab}  &\cellcolor[HTML]{1d53ab} &  & & & & 5& \\ 
        \hhline{~*{10}{-}}
        AV & &\cellcolor[HTML]{d64238} & \cellcolor[HTML]{d64238} & \cellcolor[HTML]{d64238}  &\cellcolor[HTML]{d64238} &  & & & & 6\\ 
        \hhline{~*{10}{-}}
    \end{tabular}
    \vspace{0.1cm}
     \caption{$l=3$, $k=4$. Approvals in grey, AV committee in red (first bottom row), LV committee in blue (second bottom row). Columns on the right indicate whether a voter is represented by the winning committees. 
    }
    \label{tab:CC}
\end{table} 
\end{example}


\subsection{CC-guarantee}
In \citet{LacknerSkowron2019}, the \emph{CC-guarantee} of a rule $\mathcal{R}$ is used as a quantitative measure of the rule's diversity, which is defined as follows: $\kappa_{cc}(k) = \inf_{A\in \mathcal{A}}\frac{\min_{W_\mathcal{R}\in \mathcal{R}(A,k)}s_{CC}(A,W_\mathcal{R})}{\max_{W\in S_k(C)}s_{CC}(A,W)}$, where $\mathcal{A}$ is the set of all possible preference profiles and $S_k(C)$ is the set of subsets of $C$ of size $k$.
The CC-guarantee is the lowest possible CC-improvement over all elections.

\begin{theorem}\label{thm:CC-guarantee}
    The CC-guarantee of LV is 0.
\end{theorem}

It is worth noting that the CC-guarantee of even AV is better than that of LV: $\frac{1}{k}$ \citep{LacknerSkowron2019}.
For this purpose, we can restrict the notion of CC-guarantee 
to focus on the performance of the rule on restricted domains of elections, rather than all of them, by simply restricting the set $\mathcal{A}$ to the domain that we are interested in. If $D$ is a set of elections,  $\kappa_{cc}(k)(D) = \inf_{E\in D|k_E=k}\frac{\min_{W_\mathcal{R}\in \mathcal{R}(E)}s_{CC}(A_E,W_\mathcal{R})}{\max_{W\in S_k(C_E)}s_{CC}(A_E,W)}$ is the CC-guarantee restricted to elections of the domain $D$.
\begin{theorem}\label{thm:CC-guarantee-BP}
    Let $BP$ denote the set of all possible broadcasted party-list elections, then 
    $\kappa_{cc}(k)(BP) = \frac{1}{k}$ in general, and  $\kappa_{cc}(k)(BP) = \frac{\lceil\frac{k}{l}\rceil}{k}$ for all $E\in BP$ with $l_E = l$.
\end{theorem}
\begin{proof}
    As we saw above, for any broadcasted party-list election $E=( N, C, k, l, A, L)$ with $g$ parties, the CC-score of any committee $W$ chosen by LV is 
    $\sum_{i=1}^{\min(\lceil\frac{k}{l}\rceil, g)}n_i$,
    while the maximal CC-score of any committee $W \subseteq C$ is 
     $s_{CC}(\alpha\text{-CC}) = \sum_{i=1}^{\min(k, g)}n_i$.
    Hence, we need to find an election $E$ that minimizes 
    $\frac{\sum_{i=1}^{\min(\lceil\frac{k}{l}\rceil, g)}n_i}{\sum_{i=1}^{\min(k, g)}n_i}$. 
    All else being equal, if $g< k$, this fraction can never be smaller than if $g\geq k$, so for now we assume that $g\geq k$
    : we have to minimise 
    $\frac{\sum_{i=1}^{\lceil\frac{k}{l}\rceil}n_i}{\sum_{i=1}^{k}n_i}$.
    If $l$ is not given, $l=k$ gives the lowest value: 
    $\frac{n_1}{\sum_{i=1}^{k}n_i}$, 
    and to minimize this we should minimize 
    $n_1$. Let 
    all parties except $P_1$ have $x$ voters, and $P_1$ has $x+1$ voters. Then 
    $\frac{n_1}{\sum_{i=1}^{k}n_i} = \frac{x+1}{x+1+(k-1)x}$, 
    and if $x$ goes to infinity, this will become $\frac{1}{k}$.
    If $l$ is given, 
    $\frac{\sum_{i=1}^{\lceil\frac{k}{l}\rceil}n_i}{\sum_{i=1}^{k}n_i}$ 
    is smallest when the last $\lceil\frac{k}{l}\rceil +1$ to $k$ parties are relatively as large as possible
    . 
    The first $\lceil\frac{k}{l}\rceil$ parties have to be at least 1 voter larger in order to be ranked first
    , so we minimise the fraction if 
    they have exactly 1 more voter: 
    $\frac{\sum_{i=1}^{\lceil\frac{k}{l}\rceil}n_i}{\sum_{i=1}^{k}n_i} = \frac{(x+1)\lceil\frac{k}{l}\rceil}{(x+1)\lceil\frac{k}{l}\rceil + x(k-\lceil\frac{k}{l}\rceil)}$.
    This fraction is smallest when $x$ becomes very large, 
    which 
    gives
    $\frac{\lceil\frac{k}{l}\rceil}{k}$.
\end{proof}

The CC-guarantee of AV restricted to the domain of broadcasted party-list elections is, again, $\frac{1}{k}$.

The analysis we presented for party-list profiles based on CC-improvement and CC-guarantee can be extended, obtaining similar results, to a generalization of Definition \ref{def:bpl}: broadcasted laminar profiles. We report on such generalization in Appendix \ref{app:laminar}).

\section{Proportionality} \label{sec:prop}
The above section has shown that imposing a limit on ballots can be a crude but, in some circumstances, effective mechanism to improve diversity in AV-based elections. It is natural to ask whether LV can deliver improvements also from the point of view of proportionality. 

In this section we first establish that LV does not satisfy any of the main proportionality axioms, not even on broadcasted party-list elections. We then use the proportional approval voting (PAV \cite{thiele1895}) score to compare the proportionality of LV and AV, via a notion of PAV-improvement. 

\paragraph{Proportionality axioms} 
For space reasons, we moved to Appendix \ref{app:proofs} all definitions upon which the following proposition is based, as well as its proof.
\begin{proposition}\label{prop:axioms}
    LV  fails justified representation, even on broadcasted party-list elections, and therefore fails also extended justified representation, proportional justified representation, priceability, and lower quota.
\end{proposition}

\paragraph{PAV-improvement}
Axioms are all-or-nothing conceptualizations of proportionality. We consider now also a more fine-grained approach to measure how proportional a committee is: the \emph{PAV-score} \citep{thiele1895,Lackner2023}. The PAV-score of a committee $W$ given a preference profile $A$ is defined by $s_{PAV}(A,W)=\sum_{i\in N}\left(\sum_{j=1}^{|W\cap A_i|}\frac{1}{j}\right)$. Intuitively, this is higher when candidates that many people approve are in the winning committee, but gives less weight to voters who already have comparatively more approved candidates elected. Analogously to the CC-
 improvement, we define the 
 PAV-improvement of LV as follows:
\begin{definition}[PAV-
improvement]
    For an election $E=( N, C, k, l, A, L)$, the 
    improvement in PAV-score of LV with respect to AV 
    is 
    $
    \Imp_{PAV}(E) = \frac{\min_{W_{LV}\in LV(E)}s_{PAV}(A, W_{LV})}{\max_{W_{AV}\in AV(E)}s_{PAV}(A, W_{AV})}.
    $
\end{definition}

It can be shown that in general elections, the PAV-improvement can either be below or above $1$ (Example \ref{ex:PAV_general} in Appendix \ref{app:additional}). We then focus again on the special case of broadcasted party-list elections. 
In such elections, the PAV-score of any committee elected by AV is 
$s_{PAV}(A,W_{AV})= \sum_{j=1}^{s-1}n_j(1+\frac{1}{2}+\cdots +\frac{1}{|P_j|})+n_s(1+\frac{1}{2}+\cdots +\frac{1}{k-\sum_{i=1}^{s-1}|P_i|})$, where $s$ is the largest integer $t$ such that $\sum_{i=1}^{t}|P_i|\leq k$. 
The score of a committee elected by LV is 
$s_{PAV}(A,W_{LV})= n_1(1+\frac{1}{2}+\cdots +\frac{1}{l}) + n_2(1+\frac{1}{2}+\cdots +\frac{1}{l}) +\cdots$, 
until $W$ is filled. 
Usually, this will give a 
PAV-improvement greater than $1$, but if $P_1$ is very 
popular in comparison to the other parties, 
$n_1\sum_{j=l+1}^k\frac{1}{j}$
might be larger than 
$n_2\sum_{j=1}^l\frac{1}{j} + n_3\sum_{j=1}^l\frac{1}{j}+\cdots$.


So for broadcasted party-list elections, if the relative difference between the size of the first party and the other parties is large, the PAV-
improvement is below $1$. If, on the other hand, the parties 
have similar sizes, 
the PAV-
improvement will be 
above $1$. 
\begin{proposition}\label{prop:PAV-improvement_BP}
    Let $E = ( N, C, k, l, A, L)$ be any broadcasted party-list election with parties $P_1, \ldots, P_g$. If 
    $P_1$ has at least $k$ 
    candidates,
    $E$'s PAV-improvement $\Imp_{PAV}(E)$ is
    \[
        \frac{\sum_{i = 1}^{\min(\lfloor\frac{k}{l}\rfloor, g)} \left(n_i\sum_{j=1}^l\frac{1}{j})\right) + n_{\lceil\frac{k}{l}\rceil}\sum_{j=1}^{k \mod l}\frac{1}{j}}{\sum_{j=1}^{s-1}n_j(1+\frac{1}{2}+\cdots +\frac{1}{|P_j|})+n_s\left(1+\frac{1}{2}+\cdots +\frac{1}{k-\sum_{i=1}^{s-1}|P_i|}\right)},
    \]
    where $s$ is the largest integer $t$ such that $\sum_{i=1}^{t}|P_i|\leq k$.
\end{proposition}
\begin{proof}
    This follows directly from the PAV-scores of winning committees by AV and LV in such elections.\footnote{The right term of the numerator disappears when: i) $k\mod l = 0$ since it sums from $j=1$ to $j=0$, that is, $k$ is divisible by $l$ and all candidates from $\frac{k}{l}$ parties are elected; or ii) $g<\lfloor\frac{k}{l}\rfloor$ since $P_{\lceil\frac{k}{l}\rceil}$ does not exist.}
\end{proof}
In the case where $l=\lceil \frac{k}{2}+1\rceil$, which is the ballot limit used in the company's election that motivated this study, it follows that:
\begin{corollary}\label{cor:PAV-improvement_BP}
For any broadcasted party-list election $E$ where 
$|P_1| \geq k$ and $l=\lceil \frac{k}{2}+1\rceil$, $\Imp_{PAV}(E) =
\frac{n_1\sum_{j=1}^l\frac{1}{j} + n_2\sum_{j=1}^{l-2}\frac{1}{j}}{n_1\sum_{j=1}^k\frac{1}{j}}.
$
\end{corollary}
This means that, if the most popular party is much larger (in number of voters) than the second party, the PAV-
improvement is smaller than $1$, but if 
the first two parties are roughly similar in size,
LV gives a more proportional committee than AV. 

Observe at this point that carrying out a PAV-guarantee analysis analogous to the one we developed for the CC-guarantee poses extra challenges as a closed-form expression for the PAV-score (even on broadcasted party-list elections) is hard to obtain. We will come back to the issue of proportional representation in LV in Section \ref{sec:games}.

\begin{remark}
    Even though our analysis has focused on ballot limitations in AV, because of its widespread use in practice, a similar analysis can be carried out for other ABC-voting rules. We provide some observations concerning PAV and satisfaction approval voting (SAV), where all voters get $1$ point that is split equally over all candidates the voter approves, after which the committees that maximize the sum of the points are selected \cite{brams2015satisfaction}. Details on these observations are elaborated upon in Appendix \ref{app:additional}. Imposing a ballot limit on PAV (LPAV) could improve diversity, at the cost of proportionality, on broadcasted party-list elections. However, on non-party-list elections, PAV can yield higher CC-scores than LPAV. In party-list profiles, SAV is arguably not a natural voting rule to use, as parties are penalized by their size (number of candidates), regardless of the number of voters, resulting in diversity loss. Imposing a ballot limit to SAV (LSAV) can however help mitigate this effect.
\end{remark}

\section{Computational experiments 
} \label{sec:experiment}
In many real-world elections, profiles will not be party-list, even though there might be a rough division of voters into parties. To illustrate: in the work council elections of the health insurance company that motivated this study,
the employees are roughly divided over three office locations. Every employee can vote for any candidate regardless of their location, but voters might have a tendency to vote for candidates from their own locations.
In this section, we perform computational simulations to study the extent to which the analytical results obtained for LV in broadcasted party-list elections persist in elections that are \textit{almost} broadcasted party-list. 
To generate such elections, we use the \textit{disjoint} model from \citet{szufa_2022}. The model has a parameter $\phi$ to determine the distance to a party-list profile. We then use a fully crossed simulation design to study the influence of the value of $\phi$ on the CC- and PAV-improvements, and the effect of the number of candidates $k$, the size of the votes $l$, and the number of parties $g$ on this. Below, we report only on results concerning the CC-improvement. We performed the same experiments and analysis for the PAV-improvement. The results are substantially similar and are therefore only reported in Appendix \ref{app:experiment_results}. The code of our experiments is available at \url{https://github.com/MaaikeLos/Limited_Voting/}.

\subsection{Experiments design}

\paragraph{Generating approval profiles}
The disjoint model to generate election profiles works as follows \citep{szufa_2022}. Let $p,\phi\in [0,1]$, and let $g$ be a 
positive 
integer, determining the number of parties. 
Distribute the candidates over the parties by drawing a partition of $C$ into $g$ sets, $C_1, ..., C_g$ uniformly at random. Assign every voter to a party by choosing one of these sets $C_i$ uniformly at random and let the voter approve exactly the candidates in $C_i$.\footnote{It is worth noting that this process leads to parties of approximately the same number of voters, the situation for which our theoretical analysis gave the most positive results for LV. For completeness, we also conducted the experiment with a random partition of voters over the parties, which is discussed in Appendix \ref{app:experiment_results}. Results carry over to this settings although effects are, as expected, smaller.}
We now have a party-list profile that we will modify in the following way. For every voter $v_j$ and every candidate $c\in C$ do the following: with probability $1-\phi$ keep the approval of $c$ by $v_j$ as it is, and with probability $\phi$ let $c$ be approved by $v_j$ with probability $p$. 
The above procedure generates an approval profile that may differ from a party-list profile in some votes: an \textit{almost party-list profile}. Note that if $\phi=1$, all votes will be resampled so we get the $p$-impartial culture model
, and if $\phi = 0$, we get a party-list profile. 
In our experiments we keep $p=0.5$.
\paragraph{Generating ballot profiles}
Given an almost party-list approval profile $A$, we generate a corresponding ballot profile $L$ 
with ballots of size $l$, 
with the condition that for every voter $v_i$, if $l \leq |A_i|$, then $L_i\subseteq A_i$, and if $l>|A_i|$, $A_i\subset L_i$.  For comparison with our analytical results, 
we consider elections with a broadcasting order $\succ$, but
since the approval profile may deviate from the original party-list profile, we also let $\succ$ deviate from the original 
order with the same 
probability $\phi$ that is used in generating 
the almost party-list profile. For this, we use the Mallows $\phi$-model \citep{Mallows1957} for pairwise comparisons, 
implemented as 
the Repeated Insertion Model 
\citep{Lu2014}. To generate the ballot profile, we take for every voter $v_i$ the top-$l$ candidates from $\succ$ that are in $A_i$. If 
$|A_i|<l$, we fill the vote with non-approved candidates according to $\succ$.

\paragraph{Determining winning committees and scores}
Given an election with an approval profile and ballot profile generated as described above, we calculate the winning committees for AV and LV. We make both rules resolute by using random tie-breaking. We then compute the CC-score 
of both winning committees.  Note that with resolute rules, the improvement definitions do not rely on minima or maxima: 
$\Imp_{CC}(E) = \frac{s_{CC}(A, LV(E))}{s_{CC}(A, AV(E))}.$

\paragraph{Experiments}
We run the described procedure for different values of the relevant parameters. We set $n=1500$, $m=24$, $p=0.5$, and we vary the other parameters: $\phi \in\{0, 0.05, 0.1, 0.15, 0.2,  0.25, 0.5, 0.75, 1\}$, $g\in\{2, 6, 20\}$, $k\in \{8, 16, 12\}$, and for every $k$, $l\in\{1, \frac{k}{2}, k\}$. For every combination of these parameter values, we execute the procedure 2000 times and for every 
election calculate the CC-
improvement $\Imp_{CC}$.

\subsection{Results}\label{sec:results}
Figure \ref{fig:CC-improvement} shows a boxplot of the CC-improvement for different values of $\phi$. A more refined plot with all parameter combinations separated can be found in Appendix \ref{app:experiment_results}.
Since for 
$\phi \geq 0.25$ the improvements are all close to $1$, they are merged in the plot. 
\begin{figure}[t]
    \centering
    \includegraphics[scale = 0.5]{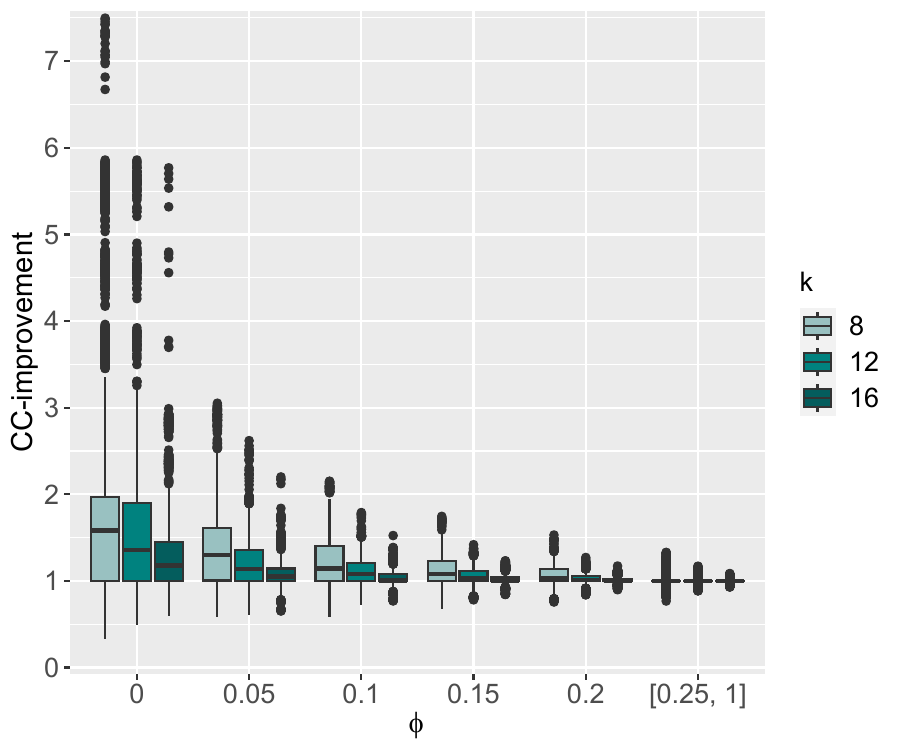}
    \caption{Boxplots of $\Imp_{CC}$ for different values of $\phi$ (dispersion). 
    \vspace{15pt} 
    }
    \label{fig:CC-improvement}
\end{figure}
In the plot we see that for $\phi=0$, LV performs better than AV. This confirms our analytical results that in broadcasted party-list elections LV provides more diversity (and proportionality) 
than AV. We also see that for $\phi\geq 0.25$,
CC-improvement 
gets very close to $1$
, so there is no real difference in score between AV and LV. 
For $0< \phi < 0.25 $, the CC-improvement  decreases as $\phi$ increases. This suggests that, for LV to be effective in improving diversity, profiles need to be close to a party-list structure: the closer they are, the more effective is the use of a ballot limit to determine gains in diversity.
Statistical analysis on the effect of the other variables on the CC-
improvement, which can be found in detail in Appendix \ref{app:variable_analysis}, indicates that for party-list elections (with $\phi=0$), $l$ and $k$ explained most of the variance in the data.

\section{Limited Voting Games}\label{sec:games}

Broadcasted party-list elections played an important role in the analytical results of Sections \ref{sec:div} and \ref{sec:prop}. In the previous section we used computational methods to obtain insights into the behavior of LV under relaxations of such assumption. In this section we again relax such assumption, but now using game-theoretic methods. Specifically, we relax the broadcasting order assumption built into the notion of broadcasted party-list elections (Definition \ref{def:bpl}) and study the strategies that parties could best use to communicate to their base about which candidates to vote for in ballots. Our analysis will make use of standard game-theoretic notions (specifically, dominant strategy, best response and pure Nash equilibrium), which, however, we refrain from defining here for space reasons.

\begin{definition}[LV-game] 
Let a party-list election frame $F = ( N, C, k, l, A)$ be given. An LV-game (for $F$) is a tuple $G = (\mathbf{P}, ( S_1, ..., S_g), O, f, ( U_1, ..., U_g ) )$ where:\footnote{We base our notation on \cite{Bonanno2018}.}
\begin{itemize}
    \item $\mathbf{P} = \{P_1, ..., P_g\}$, players are the parties of $F$.
    \item For every $P_i\in \mathbf{P}, S_i$ is the set of strategies of party $P_i$, consisting of tuples of the form $s_i = ((B_1,h_1), \ldots, (B_{t_i},h_{t_i})) $ with $\sum_{1\leq x\leq t_i}h_x = n_i$, where each $B_x$ is a limited ballot (a set of $l$ candidates from $C$), $t_i$ is the number of different such ballots expressed by voters of $P_i$, and each $h_x$ is a natural number denoting the number of voters in party $P_i$ that will vote $B_x$. We denote by $\mathbf{S} = S_1 \times ... \times S_g$ the set of strategy profiles, where every $S\in \mathbf{S}$ consists of one strategy for every party: $S = (s_1, ..., s_g)$.\footnote{Observe right away that every strategy profile induces a set of limited ballot profiles with identical approval counts for all candidates.}
    \item $O$ is the set of outcomes, that is, all $k$-sized subsets of the set of candidates $C$: $O=\{o\subseteq C: |o|=k\}$.
    \item $f:\mathbf{S}\rightarrow O$ determines the outcome of each strategy profile and is given by $f(S) = LV(E)$, where  $E = (F, L)$, and $L$ is (one of) the ballot profile(s) induced by $S$: for every $1\leq x\leq t_i$, $h_x$ voters of party $P_i$ vote ballot $B_x$.
    \item For every $P_i\in \mathbf{P}$, $U_i: O\rightarrow \mathbb{N}$ is the utility function of party $P_i$, which is defined as $U_i(o) = \min_{W\in o}(|W\cap P_i|)$. 
\end{itemize} 
\end{definition}
The definition deserves a few comments. First of all, notice that the strategy of each party $P_i$ consists of $t_i$ different ballots, and the number of voters that vote for each ballot, rather than consisting of a function assigning specific ballots to each party member. Given that LV is anonymous, this is without loss of generality. Clearly, if $t_i=1$ the strategy of party $P_i$ is a broadcasting strategy. If $t_i=n_i$, the strategy determines different ballots for each voter. Finally, observe that parties are assumed to be `pessimistic': they get the utility of the worst of all tied committees in the outcome.



\subsection{LV games with broadcasting strategies}
We start by studying LV-games where $t_i = 1$ for all parties $P_i$, which we refer to as broadcasted LV-games. This allows us to provide also a game-theoretic grounding to Definition \ref{def:bpl}. In such games, strategies are pairs $((B_1,h_1))$ (where $h_1=n_i$), so we will simply denote a strategy by its ballot. We can observe right away that if a party $P_i$ consists of at least $l$ candidates, then any strategy $s_i \subseteq P_i$ is weakly dominant for $P_i$. We thus establish the following results. 
\begin{proposition}
In any broadcasted LV-game, if for all parties $P_i$, $|P_i|\geq l$, then any strategy profile $S$ that satisfies $s_i \subseteq P_i$ for all $P_i$, is a pure strategy Nash equilibrium.
\end{proposition}

\begin{corollary}
Any broadcasted party-list election $( N, C, k, l, A, L)$ where for all parties $P_i$, $|P_i|\geq l$ is a pure strategy Nash equilibrium of a broadcasted LV-game for election frame $(N, C, k, l, A)$.
\end{corollary}


\subsection{Equilibria in general LV games}

In this section we study LV-games in their generality, where parties can specify for separate voters which candidates to vote for. First of all, observe that if we were to consider such games without a ballot size limit (that is, with outcomes determined by AV instead of LV), it would be dominant for parties to let their base vote for all their candidates. In equilibrium, therefore, the largest party would get all its candidates elected, then the second-largest party would, and so on untill the committee is filled (recall Example \ref{ex:introduction_example}). We will show that in LV-games, instead, a fairly strong form of proportionality can be obtained in a weaker form of equilibrium known as $\epsilon$-Nash equilibrium. A strategy profile $S$ of an LV-game is an $\epsilon$-Nash pure strategy equilibrium (with $\epsilon > 0$) if, for all parties $P_i$ and all strategies $s'_i$, $U_i(S) \geq U_i(S_{-i}, s'_i)- \epsilon$, where $S_i$ denotes the profile $S$ without $i$'s entry. Intuitively, a profile is $\epsilon$-Nash whenever no deviation exists for any agent that would yield a utility gain of more than $\epsilon$.

We recall that a committee satisfies \emph{lower quota} if every party $P_i$ gets at least $\lfloor\frac{k\cdot n_i}{|N|}\rfloor$ representatives in the committee \citep{Lackner2023, Peters2021}. 
Consider now a strategy (for the LV-game) in which a party $P_i$ (with $l\leq \lfloor\frac{k\cdot n_i}{|N|}\rfloor \leq |P_i|$) selects a set $X_i$ of $\lfloor\frac{k\cdot n_i}{|N|}\rfloor$ candidates from $P_i$ and lets its voters vote in such a way that they spread their votes as evenly as possible over the candidates in $X_i$. 
Formally:
\begin{definition}
    Let an LV-game be given.
    For any party $P_i$, $\mathcal{LQ}_i$ denotes the set of strategies $((B_1,h_1), \ldots, (B_{t_i},h_{t_i}))$, such that for every candidate $c \in C$, the number of votes $c$ gets is $y_c = \sum_{x:c\in B_x}h_x \in \{\lfloor \frac{n_i}{\lfloor\frac{k\cdot n_i}{|N|}\rfloor}\rfloor, \lceil \frac{n_i}{\lfloor\frac{k\cdot n_i}{|N|}\rfloor}\rceil, 0\}$, and $\sum_{c\in P_i}y_c = n_i\cdot l$.
\end{definition}
The following lemma shows that, if a party follows the above lower quota `heuristic' to determine their strategy in an LV-game, then the party can guarantee a lower-quota representation for itself, provided the ballot limit does not exceed the lower quota threshold.
\begin{proposition} 
Let an LV-game be given.
    If any party $P_i$ (with $l\leq \lfloor\frac{k\cdot n_i}{|N|}\rfloor \leq |P_i|$) plays a strategy in $\mathcal{LQ}_i$,
    then there exists a committee in the outcome of the game in which all candidates from $P_i$'s strategy are selected.
\end{proposition}
\begin{proof}
     Suppose a party $P_i$ plays such a strategy and $X_i$ denotes the candidates in $P_i$ which receive a vote. Every candidate in $X_i$ gets on average $\frac{n_i\cdot l}{\lfloor\frac{k\cdot n_i}{|N|}\rfloor}$ votes. However, this fraction is not necessarily integer, but every candidate gets at least $\lfloor\left(\frac{n_i\cdot l}{\lfloor\frac{k\cdot n_i}{|N|}\rfloor}\right)\rfloor$ votes.
     Since $\frac{n_i\cdot l}{\lfloor\frac{k\cdot n_i}{|N|}\rfloor} \geq \frac{n_i\cdot l}{\frac{k\cdot n_i}{|N|}}$, also $\lfloor\left(\frac{n_i\cdot l}{\lfloor\frac{k\cdot n_i}{|N|}\rfloor}\right)\rfloor \geq \lfloor \left( \frac{n_i\cdot l}{\frac{k\cdot n_i}{|N|}}\right) \rfloor = \lfloor\frac{|N|\cdot l}{k}\rfloor $.
     Suppose there would be no outcome of LV in which all candidates from $X_i$ would be elected. Then there must be another party $P_j$ that gets more than its lower quota $\lfloor\frac{k\cdot n_j}{|N|}\rfloor$ of candidates in all outcomes, so $P_j$ gets at least $\lfloor\frac{k\cdot n_j}{|N|}\rfloor+1$ candidates that each have at least $\lfloor\frac{|N|\cdot l}{k}\rfloor + 1$ votes (if one of these candidates got exactly $\lfloor\frac{|N|\cdot l}{k}\rfloor$ votes or less, this candidate would either loose from or be tied with candidates from $P_i$). In total, $P_j$ has $n_j\cdot l$ votes to spend, so if $P_j$ votes on a set of $\lfloor\frac{k\cdot n_j}{|N|}\rfloor+1$ candidates, every candidate in this set gets on average $\frac{n_j\cdot l}{\lfloor\frac{k\cdot n_j}{|N|}\rfloor+1}$ votes. If all the candidates in this set have to get \emph{at least} $\lfloor\frac{|N|\cdot l}{k}\rfloor + 1$ votes, then the floor of the average has to be at least that amount: in other words, it has to be the case that $\lfloor\left(\frac{n_j\cdot l}{\lfloor\frac{k\cdot n_j}{|N|}\rfloor+1}\right) \rfloor \geq \lfloor\frac{|N|\cdot l}{k}\rfloor + 1$.
     However, since $\lfloor\frac{k\cdot n_j}{|N|}\rfloor+1 > \frac{k\cdot n_j}{|N|}$ we have that 
     $\frac{|N|\cdot l}{k} =
     \frac{n_j\cdot l}{\frac{k\cdot n_j}{|N|}} >
     \frac{n_j\cdot l}{\lfloor\frac{k\cdot n_j}{|N|}\rfloor+1}$ and by definition of the floor we have
     $ \frac{n_j\cdot l}{\lfloor\frac{k\cdot n_j}{|N|}\rfloor+1} \geq 
     \lfloor\left(\frac{n_j\cdot l}{\lfloor\frac{k\cdot n_j}{|N|}\rfloor+1}\right) \rfloor$, so this gives us
     $\frac{|N|\cdot l}{k}\geq \lfloor\frac{|N|\cdot l}{k}\rfloor +1$, which is a contradiction. Clearly also if $P_j$ spreads its votes over more than $\lfloor\frac{k\cdot n_j}{|N|}\rfloor+1$ candidates, it cannot get at least $\lfloor\frac{k\cdot n_j}{|N|}\rfloor+1$ candidates with at least $\lfloor\frac{|N|\cdot l}{k}\rfloor + 1$ votes.
\end{proof}


\begin{corollary} \label{cor:lq}
Let an LV-game be given such that for all parties $P_i$, $l\leq \lfloor\frac{k\cdot n_i}{|N|}\rfloor \leq |P_i|$. If each party $P_i$ plays a strategy in $\mathcal{LQ}_i$, then all committees in the outcome of the game satisfy lower quota.
\end{corollary}

However, the strategy profile resulting from lower quota strategies is in equilibrium only under very restrictive conditions, namely that the lower quotas perfectly divide $k$.
\begin{proposition}\label{prop:NE_lower_quota}
Let an LV-game be given such that for all parties $P_i$, $l\leq \lfloor\frac{k\cdot n_i}{|N|}\rfloor \leq |P_i|$.
A strategy profile where each party $P_i$ plays a strategy in $\mathcal{LQ}_i$, is a pure strategy Nash equilibrium if and only if $\sum_{i\leq g}\lfloor\frac{k\cdot n_i}{|N|}\rfloor = k$.
\end{proposition}
\begin{proof}
Left-to-right. If $\sum_{i\leq g}\lfloor\frac{k\cdot n_i}{|N|}\rfloor < k$, then any party $P_i$ can best respond by  by spreading their votes over $k-\sum_{ i\leq g}\lfloor\frac{k\cdot n_i}{|N|}\rfloor + \lfloor\frac{k\cdot n_i}{|N|}\rfloor$ candidates instead of $\lfloor\frac{k\cdot n_i}{|N|}\rfloor$, because then still no more than $k$ candidates will get any votes, so all candidates that do get votes will be elected. 
Right-to-left.  If $\sum_{i\leq g}\lfloor\frac{k\cdot n_i}{|N|}\rfloor = k$, then for all parties $P_i$, $\lfloor\frac{k\cdot n_i}{|N|}\rfloor=\frac{k\cdot n_i}{|N|}$. This means that all candidates in $X_i$ get exactly $\frac{n_i\cdot l}{\frac{k\cdot n_i}{|N|}}=\frac{|N|\cdot l}{k}$ votes. If any party $P_i$ spreads its votes over more than $\frac{k\cdot n_i}{|N|}$ candidates, their candidates get on average less votes, and therefore will not be elected in the place of candidates of other parties, that still get $\frac{|N|\cdot l}{k}$ votes (some candidates of $P_i$ can still be tied with candidates from other parties, but not strictly win from them). 
Hence, no best response exists. 
\end{proof}
The left-to-right direction of the proof provides a bound on the utility that  can be gained by deviating from a lower quota strategy. 
\begin{theorem}\label{prop:delta_NE_lower_quota}
Let an LV-game be given such that for all parties $P_i$, $l\leq \lfloor\frac{k\cdot n_i}{|N|}\rfloor \leq |P_i|$. A strategy profile where each party $P_i$ plays a strategy in $\mathcal{LQ}_i$ is an $\epsilon$-Nash equilibrium where $\epsilon = k-\sum_{i\leq g}\lfloor\frac{k\cdot n_i}{|N|}\rfloor$.
\end{theorem}
So, the closer the sum of the lower quotas is to the size of the committee, the more stable the lower quota strategies become in an LV-game. This, in turn, shows that LV could, in party-list profiles, provide lower-quota proportional representation almost stably in the presence of rational parties, even though it fails the lower quota axiom (Proposition \ref{prop:axioms}). 
Examples illustrating the results of this section are provided in Appendix \ref{app:additional}.


\section{Conclusions}

We showed that in applications where elections tend to be structured around parties, restricting approval ballots to a fixed size may be an effective 'quick-and-dirty' fix of AV if one cares about obtaining better representation (both in terms of diversity and proportionality) without completely changing the voting rule.

We conclude by mentioning a few directions for future work. 
First, a variant of LV allows voters to vote on \emph{at most} $l$ candidates rather than precisely $l$. Most of our positive results will still hold for such elections, but some of the negative results may be mitigated by allowing for shorter ballots since voters are not forced anymore to vote on candidates they do not approve. It would be good to extend our analysis also to this variant of LV. 
Second, while we worked with PAV-improvement, it would be interesting to know how LV performs with respect to other representation measures, e.g. the \emph{proportionality degree} as defined in \citet{Skowron2021}.
Third, we made the common assumption that voters' preferences are binary and voters get equal satisfaction from all their approved candidates. 
A natural extension of this paper would be to representation in LV with other approval-based satisfaction functions, following \citet{brill2023}.


\newpage

\appendix

\section{Laminar profiles}\label{app:laminar}

A generalisation of 
party-list profiles are \emph{laminar profiles}, introduced in \citet{Peters2020limwelf}, where parties can consist of smaller subparties.
Laminar elections can be defined recursively for LV elections as follows. An election $E = ( N, C, k, l, A, L)$ is laminar if either: 
(1) $E$ is unanimous and $|C|\geq k$; 
(2) There is a candidate $c\in C$ such that $c\in A_i$ for all $i \in N$, the election $E_{-c}=( N, C\backslash\{c\}, k-1, l', A_{-c}, L')$, where $E_{-c}$ is $E$ once we remove $c$, i.e., $ A_{-c} =  (A_1\backslash \{c\}, ..., A_n\backslash\{c\})$, is not unanimous 
and is laminar, $l'$ and $L'$ are any limit and ballot profile compatible with $E_{-c}$; 
or (3) There are two laminar elections $E_1 = ( N_1, C_1, k_1, l_1, A_1, L_1)$ and $E_2 = ( N_2, C_2, k_2, l_2, A_2, L_2)$ with $N=N_1\cup N_2$, $C = C_1\cup C_2$, $A= A_1\cup A_2$, $k=k_1+k_2$, and $C_1\cap C_2 = N_1 \cap N_2 =\emptyset$, such that  $|N_1|\cdot k_2 = |N_2|\cdot k_1$. 
We refer to \citet{Peters2020limwelf} for examples of laminar elections and intuitions motivating the definition.

\begin{definition}[Broadcasted laminar elections]
An election $E = ( N, C, k, l, A, L)$ is a \textit{broadcasted laminar election} if $A$ is a laminar profile 
and 
$E$ is consistent with a broadcasting order $\succ_C$  such that for any $c, c'\in C$: if $s_{AV}(c)>s_{AV}(c')$, then $c\succ_C c'$ (more popular candidates are prioritized)
.
\end{definition}
The extra condition on the broadcasting order is to ensure that in broadcasted laminar elections candidates with more approvals are voted before candidates with less approvals.
Table \ref{tab:example-broadcasted} gives an example of a broadcasted laminar profile.

\begin{table}[h]
    \caption{
    A laminar election ($l=2$) with a broadcasting order $c_1\succ_C c_2 \succ_C c_3 \succ_C c_4 \succ_C c_5 \succ_C c_6$ (left), and a non-broadcasted election on the same approval profile (right). Approvals in grey, votes indicated by x.}
    \label{tab:example-broadcasted}
    \small
    \renewcommand{\arraystretch}{0.7}
    \setlength{\tabcolsep}{3pt}
    \centering
\begin{tabular}{ cc }   
        \begin{tabular}{c|c|c|c|c|c|c|}
           \multicolumn{1}{c}{} & \multicolumn{1}{c}{$c_1$} & \multicolumn{1}{c}{$c_2$} & \multicolumn{1}{c}{$c_3$} & \multicolumn{1}{c}{$c_4$} & \multicolumn{1}{c}{$c_5$} & \multicolumn{1}{c}{$c_6$} \\
            \hhline{~------} 
            $v_1$ & \cellcolor[HTML]{C1C1C1} x & \cellcolor[HTML]{C1C1C1} x & \cellcolor[HTML]{C1C1C1}  &  &  & \\ \hhline{~------} 
            $v_2$ & \cellcolor[HTML]{C1C1C1} x & \cellcolor[HTML]{C1C1C1} x & \cellcolor[HTML]{C1C1C1}  &  &  & \\ \hhline{~------} 
            $v_3$ & \cellcolor[HTML]{C1C1C1} x &  & & \cellcolor[HTML]{C1C1C1} x & \cellcolor[HTML]{C1C1C1}  &\cellcolor[HTML]{C1C1C1} \\ \hhline{~------} 
        $v_4$ & \cellcolor[HTML]{C1C1C1} x & &  & \cellcolor[HTML]{C1C1C1} x & \cellcolor[HTML]{C1C1C1}  & \\ \hhline{~------} 
        \end{tabular} &  
        \begin{tabular}{c|c|c|c|c|c|c|}
           \multicolumn{1}{c}{} & \multicolumn{1}{c}{$c_1$} & \multicolumn{1}{c}{$c_2$} & \multicolumn{1}{c}{$c_3$} & \multicolumn{1}{c}{$c_4$} & \multicolumn{1}{c}{$c_5$} & \multicolumn{1}{c}{$c_6$} \\
                \hhline{~------}  
            $v_1$ & \cellcolor[HTML]{C1C1C1} x & \cellcolor[HTML]{C1C1C1}  & \cellcolor[HTML]{C1C1C1} x &  &  & \\ \hhline{~------}  
            $v_2$ & \cellcolor[HTML]{C1C1C1}  & \cellcolor[HTML]{C1C1C1} x & \cellcolor[HTML]{C1C1C1} x &  &  & \\ \hhline{~------}  
            $v_3$ & \cellcolor[HTML]{C1C1C1}  &  & & \cellcolor[HTML]{C1C1C1}  & \cellcolor[HTML]{C1C1C1}x  &\cellcolor[HTML]{C1C1C1}x \\ \hhline{~------}  
            $v_4$ & \cellcolor[HTML]{C1C1C1} x & &  & \cellcolor[HTML]{C1C1C1} x & \cellcolor[HTML]{C1C1C1}  & \\ \hhline{~------}  
        \end{tabular}   \\
\end{tabular}
\end{table}

In the main paper we found that 
in more aligned elections (of which broadcasted party-list elections are an extreme case), the CC-
improvement is 
higher. 
This raises the question of whether in slightly less aligned elections, such as laminar elections, this is still the case.

\begin{proposition}
\label{prop:CCgain-brlaminar}
    In any broadcasted laminar election $E = ( N, C, k, l, A, L)$ where AV and LV are resolute
    , if for all $i\in N$, $|A_i|\geq l$, 
    $\Imp_{CC}\geq 1.$      
\end{proposition}

\begin{proof}
    Suppose towards a contradiction that there is such broadcased laminar election $E$ where $W_{AV}$ is the winning committee of AV and $W_{LV}$ 
    that of LV, in which  $s_{CC}(A, W_{AV})>s_{CC}(A,W_{LV})$. 
    Then there must be a voter $v$ who has at least one of her approved candidates in $W_{AV}$, but none in $W_{LV}$.  
    Therefore, even the first candidate (let's call it $c_1$) on which $v$ will vote (the candidate in $A_v$ 
    ranked highest in $\succ_C$) is not included in $W_{LV}$.
    Since $E$ is laminar, and $\succ_C$ is the same for all voters, and $l\geq 1$, all voters that approve $c_1$ must vote on it, so $s_{LV}(c_1) = s_{AV}(c_1)$. We also know that there is at least one $c'\in A_v\cap W_{AV}$. Since $c_1 \succ_C c'$, we know that $s_{AV}(c_1)\geq s_{AV}(c')$. If $s_{AV}(c_1) > s_{AV}(c')$ we can assume $c_1 \in W_{AV}$, so both then and when $s_{AV}(c_1) = s_{AV}(c')$, there exists a candidate  $c''\in A_v\cap W_{AV}$ with $s_{AV}(c'') = s_{AV}(c_1) = s_{LV}(c_1)$. However, since we have that  for all $i\in N$, $|A_i|\geq l$, we know that for all candidates $c\in C, s_{AV}(c)\geq s_{LV}(c)$. Therefore, the threshold of approvals to be elected in $W_{AV}$ must be at least as high as the threshold of votes to be elected in $W_{LV}$, and hence if the value of $s_{LV}(c_1)$ is not enough to be elected by LV, the value of $s_{AV}(c'') = s_{LV}(c_1)$ can also not be enough to be elected by AV. 
    Hence, $c''$ cannot be elected by AV: a contradiction. 
\end{proof}

The assumption that AV and LV are resolute in $E$ and the assumption of a broadcasting order are necessary for this result, see Examples \ref{ex:nonresolute} and \ref{ex:laminar_wrong_order}. 

\begin{example}\label{ex:nonresolute}
    The condition that AV and LV are resolute in 
    the election in the proof of 
    Proposition \ref{prop:CCgain-brlaminar} 
    is necessary to prevent the possibility where $s_{LV}(c_1)$ was  enough to be elected by LV, but just because of tie-breaking rules was not elected. We give an example of such election in Table \ref{tab:nonresolute}: AV has a maximal CC-score of 2 with for example $\{c_3, c_4\}$, while LV has a minimal CC-score of 1 with for example $\{c_4, c_5\}$. 
    Hence, even though this is a broadcasted laminar election where for all $i\in N$, $|A_i|\geq l$, $Improvement_{CC}=\frac{1}{2}\not \geq 1.$ 
\begin{table}[t]
    \caption{ Example with AV and LV non-resolute and 
    $Improvement{CC}=\frac{1}{2}$. Approvals in grey, limited votes indicated by x, $k=l=2$. The bottom rows show the winning committees according to LV (blue) and AV (red), the columns on the right indicate whether voters are represented by the winning committees.}
    \label{tab:nonresolute}
    \centering
    \small
    \renewcommand{\arraystretch}{0.8}
    \setlength{\tabcolsep}{3pt}
        \begin{tabular}{c|c|c|c|c|c|c||c|c|}
        \multicolumn{1}{c}{} & \multicolumn{1}{c}{} &\multicolumn{1}{c}{} &\multicolumn{1}{c}{} &\multicolumn{1}{c}{} &\multicolumn{1}{c}{} &\multicolumn{1}{c}{} &\multicolumn{2}{c}{repr.} \\
           \multicolumn{1}{c}{} & \multicolumn{1}{c}{$c_1$} & \multicolumn{1}{c}{$c_2$} & \multicolumn{1}{c}{$c_3$} & \multicolumn{1}{c}{$c_4$} & \multicolumn{1}{c}{$c_5$} & \multicolumn{1}{c}{$c_6$} & \multicolumn{1}{c}{LV} & \multicolumn{1}{c}{AV} \\
            \hhline{~--------}
            $v_1$ & \cellcolor[HTML]{C1C1C1} x & \cellcolor[HTML]{C1C1C1} x & \cellcolor[HTML]{C1C1C1}  &  &  & & & \cellcolor[HTML]{d64238} \\ \hhline{~--------}
            $v_2$ &  &  &   &  \cellcolor[HTML]{C1C1C1} x& \cellcolor[HTML]{C1C1C1} x & \cellcolor[HTML]{C1C1C1} & \cellcolor[HTML]{1d53ab}&\cellcolor[HTML]{d64238}\\ \hline 
        \multicolumn{1}{c}{} & \multicolumn{6}{c}{winning committees}& \multicolumn{2}{c}{CC-score} \\   
        \hline
            LV &  &  &   &  \cellcolor[HTML]{1d53ab}& \cellcolor[HTML]{1d53ab}& & 1 &\\ \hhline{~--------}
            AV &  &  &  \cellcolor[HTML]{d64238} &  \cellcolor[HTML]{d64238}& & & & 2\\ \hhline{~--------}
        \end{tabular}
\end{table}
\end{example}

\begin{example}\label{ex:laminar_wrong_order}
In an election with a laminar profile that does have a broadcasting order but not one that starts with the most popular candidates, the CC-improvement can be below 1, as opposed to 
Proposition \ref{prop:CCgain-brlaminar}.
Take
the laminar profile in 
Table \ref{tab:smaller}. 
With LV where candidates are voted according to lexicographic order, $v_3, v_4,$ and $v_5$ will vote for $c_3, c_4,$ and $c_5$ respectively instead of combining their forces on $c_6$ and $c_7$. This makes any outcome of LV leaving one of $v_3, v_4,$ and $v_5$ without representative, while any outcome of AV would elect at least one candidate from the approval set of every voter.

\begin{table}[t]
    \caption{Laminar election with (lexicographic) broadcasting order that does not prioritize the most popular candidates, $k=3$ and $l=1$. Approvals in grey, votes indicated by x, the bottom rows show possible winning committees according to LV (blue) and AV (red) and the columns on the right indicate whether a voter is represented by the committees. 
    }
    \label{tab:smaller}
    \small
    \centering
    \renewcommand{\arraystretch}{0.8}
    \setlength{\tabcolsep}{3pt}
        \begin{tabular}{c|c|c|c|c|c|c|c||c|c|}
            \multicolumn{1}{c}{}& \multicolumn{1}{c}{}& \multicolumn{1}{c}{}& \multicolumn{1}{c}{}& \multicolumn{1}{c}{}& \multicolumn{1}{c}{}& \multicolumn{1}{c}{}& \multicolumn{1}{c}{}&\multicolumn{2}{c}{repr.} \\
                    \multicolumn{1}{c}{}& \multicolumn{1}{c}{ $c_1$} & \multicolumn{1}{c}{$c_2$} & \multicolumn{1}{c}{$c_3$} & \multicolumn{1}{c}{$c_4$} & \multicolumn{1}{c}{$c_5$} & \multicolumn{1}{c}{$c_6$} & \multicolumn{1}{c}{$c_7$}& \multicolumn{1}{c}{LV} & \multicolumn{1}{c}{AV}  \\
                    \hhline{~*{9}{-}}
            $v_1$ & \cellcolor[HTML]{C1C1C1} x & \cellcolor[HTML]{C1C1C1}  &  &  & &  &  &\cellcolor[HTML]{1d53ab} &  \cellcolor[HTML]{d64238}\\ \hhline{~*{9}{-}}
            $v_2$ & \cellcolor[HTML]{C1C1C1} x & \cellcolor[HTML]{C1C1C1}  &  &  & &  &  &\cellcolor[HTML]{1d53ab} &  \cellcolor[HTML]{d64238}\\ \hhline{~*{9}{-}}
            $v_3$ &  &  & \cellcolor[HTML]{C1C1C1} x & &  &\cellcolor[HTML]{C1C1C1} &\cellcolor[HTML]{C1C1C1} & \cellcolor[HTML]{1d53ab} &  \cellcolor[HTML]{d64238}\\ \hhline{~*{9}{-}}
            $v_4$ &  &  & & \cellcolor[HTML]{C1C1C1} x  &  &\cellcolor[HTML]{C1C1C1} &\cellcolor[HTML]{C1C1C1} & \cellcolor[HTML]{1d53ab} &  \cellcolor[HTML]{d64238}\\ \hhline{~*{9}{-}}
            $v_5$ &  &  & & &  \cellcolor[HTML]{C1C1C1} x &\cellcolor[HTML]{C1C1C1} &\cellcolor[HTML]{C1C1C1}& & \cellcolor[HTML]{d64238}\\  \hline 
            \multicolumn{1}{c}{} & \multicolumn{7}{c}{winning committees}& \multicolumn{2}{c}{CC-score} \\ \hline
            LV & \cellcolor[HTML]{1d53ab} &  & \cellcolor[HTML]{1d53ab}  &\cellcolor[HTML]{1d53ab} &  & & & 4 & \\ \hhline{~*{9}{-}}
            AV &\cellcolor[HTML]{d64238}  & &  & & &\cellcolor[HTML]{d64238}   & \cellcolor[HTML]{d64238} & & 5\\ \hhline{~*{9}{-}}
        \end{tabular}
\end{table} 
\end{example}

The CC-guarantee of LV in broadcasted laminar elections is the same as that in broadcasted party-list elections:
\begin{proposition}\label{prop:CC-guarantee_BL}
    Let $BL$ denote the set of all possible broadcasted laminar elections, then 
    the CC-guarantee of LV restricted to BL is
    $\kappa_{cc}(k)(BL) = \frac{1}{k}$ in general, and  $\kappa_{cc}(k)(BL) = \frac{\lceil\frac{k}{l}\rceil}{k}$ for all $E\in BL$ with $l_E = l$.
\end{proposition}
\begin{proof}
   Observe first that since broadcasted party-list elections are a subset of broadcasted laminar elections, the CC-guarantee of any voting rule can never be higher in broadcasted laminar elections than in broadcasted party-list elections.
   We show that for LV, the CC-guarantee in broadcasted laminar elections is also not lower than that in broadcasted party-list elections.
   Take any broadcasted laminar election $E$ and let $g$ be the number of superparties in $E$, i.e. the lowest number $g$ such that there exists a partition of $V$ into $g$ disjoint sets of voters where every set has at least one candidate that all the voters in the set approve. 
   Just as in party-list profiles, we order the superparties by number of voters: $P_1$ is the most popular, $P_2$ the second-most, etc. The maximum CC-score of any committee in $E$ is $s_{CC}(\alpha\text{-CC}) = \sum_{i=1}^{\min(k, g)}|P_i|$, which we can obtain by taking one (in that party unanimously approved) candidate from every superparty, until all voters are satisfied or the committee is filled. 
   The minimal CC-score of the outcome of LV over all broadcasted laminar profiles is the score of a party-list profile. To make the CC-score minimal, we want as few voters to be represented as possible. For the largest superparties (the ones that contain the most popular candidates), all voters will be satisfied anyway, since the most popular candidate is chosen first by LV. Hence, to minimise the CC-score, we need as many candidates as possible from the largest parties, since then we need less candidates that might represent other voters. Then, if for any superparty, a non-unanimous candidate of that party is elected by LV, that means that already the subparty that approves it has enough votes to elect a candidate. And that implies that if the unanimous candidates of the superparty were split into two candidates, one with the subparty as approvers and one with the rest of the voters of the superparty, at least as many candidates from the superparty could have been chosen, with at most as many voters being represented. But that implies that if the profile were party-list, at least as many candidates from the superparty could have been chosen, with at most as many voters being represented, so the CC-score in any broadcasted laminar profile is at least that of a party-list profile. Hence, the CC-guarantee of broadcasted laminar elections cannot be lower than that of broadcasted party-list elections.
\end{proof}


\section{Additional Examples
}\label{app:additional}

\subsection{Pareto inefficiency of LV and broadcasting order}

\begin{example}\label{ex:pareto}
     LV does not satisfy Pareto efficiency.
    \begin{table}[t]
        \caption{$l=3$, $k=5$ (the committee size). Approvals in grey, limited votes indicated by x. Some voters approve less that $l$ candidates, then we let them fill their ballot lexicographically. 
        In the bottom row, the winning committee according to LV is indicated in blue.
        }
        \label{tab:pareto-counter}
        \renewcommand{\arraystretch}{0.8}
        \setlength{\tabcolsep}{3pt}
        \small
            \centering
            \begin{tabular}{c|c|c|c|c|c|c|c|c|c|c|}
                  \multicolumn{1}{c}{}& \multicolumn{1}{c}{
                  $c_1$} & \multicolumn{1}{c}{$c_2$} & \multicolumn{1}{c}{$c_3$} & \multicolumn{1}{c}{$c_4$} & \multicolumn{1}{c}{$c_5$} & \multicolumn{1}{c}{$c_6$} & \multicolumn{1}{c}{$c_7$} & \multicolumn{1}{c}{$c_8$} & \multicolumn{1}{c}{$c_9$} & \multicolumn{1}{c}{$c_{10}$}  \\
                        \hhline{~*{10}{-}}
                $v_1$ & \cellcolor[HTML]{C1C1C1} x & \cellcolor[HTML]{C1C1C1} x &  &  &  &\cellcolor[HTML]{C1C1C1} x &\cellcolor[HTML]{C1C1C1}  & \cellcolor[HTML]{C1C1C1}  & &\\ \cline{2-11}
                $v_2$ & x &  &   & \cellcolor[HTML]{C1C1C1} x &  &  &\cellcolor[HTML]{C1C1C1} x & & & \\ 
                \hhline{~*{10}{-}}
                $v_3$ &  x &  x & &  & & & & \cellcolor[HTML]{C1C1C1} x &  & \\ 
                \hhline{~*{10}{-}}
                $v_4$ & &  &\cellcolor[HTML]{C1C1C1} x & \cellcolor[HTML]{C1C1C1} x & \cellcolor[HTML]{C1C1C1} x & & &\cellcolor[HTML]{C1C1C1} & \cellcolor[HTML]{C1C1C1} & \cellcolor[HTML]{C1C1C1}  \\ 
                \hhline{~*{10}{-}}
               $v_5$ & &  &\cellcolor[HTML]{C1C1C1} x & \cellcolor[HTML]{C1C1C1} x & \cellcolor[HTML]{C1C1C1} x & & &\cellcolor[HTML]{C1C1C1} & \cellcolor[HTML]{C1C1C1} & \cellcolor[HTML]{C1C1C1}  \\ \hline \hline
               LV & \cellcolor[HTML]{1d53ab}  & \cellcolor[HTML]{1d53ab}   &\cellcolor[HTML]{1d53ab}  & \cellcolor[HTML]{1d53ab}  & \cellcolor[HTML]{1d53ab}  & & & & &  \\ 
               \hhline{~*{10}{-}}

            \end{tabular}
    \end{table}
    In the profile in Table \ref{tab:pareto-counter}, the committee $W=\{c_1, c_2, c_3, c_4, c_5\}$ wins according to LV: its candidates have the highest number of votes. However, the committee $W'= \{c_6, c_7, c_8, c_9, c_{10}\}$ dominates it: no voter has less approved candidates in $W'$ than in $W$, and $v_1$ and  $v_3$ have more approved candidates in $W'$. 
    Note that approval voting does indeed satisfy Pareto efficiency for dichotomous preferences \citep{Lackner2023}.
\end{example}

\begin{example}\label{ex:non-strategic}
Broadcasting unique orders to voters might not be the optimal choice for parties.
   Take an election with $k=6$, $l=2$ with a party-list profile $A$ such that there are two parties, $a$ and $b$ that both have 6 voters. See Table \ref{tab:example_non_strategic}. Party $a$ broadcasts to all its voters to vote on the same two candidates $a_1$ and $a_2$, party $b$ is more coordinated and asks two voters to vote for $b_1$ and $b_2$, two voters to vote for $b_3$ and $b_4$, and two voters to vote for $b_5$ and $b_6$. In this way $a$ gets only 2 seats in the winning committee, while $b$ gets 4 seats, so clearly it is not the most strategic choice to ask all voters to vote for the same candidates.
   \begin{table}[t]
   \caption{$l=2$, $k=6$. Approvals in grey, limited votes indicated by x. A winning committee of LV is presented at the bottom row.}
    \label{tab:example_non_strategic}
    \small
    \renewcommand{\arraystretch}{0.8}
    \setlength{\tabcolsep}{3pt}
        \centering
        \begin{tabular}{c|c|c|c|c|c|c|c|c|c|c|}
                    \multicolumn{1}{c}{}& \multicolumn{1}{c}{$a_1$} & \multicolumn{1}{c}{$a_2$} & \multicolumn{1}{c}{$a_{...}$}& \multicolumn{1}{c}{$b_1$} & \multicolumn{1}{c}{$b_2$} & \multicolumn{1}{c}{$b_3$} & \multicolumn{1}{c}{$b_4$} & \multicolumn{1}{c}{$b_5$} & \multicolumn{1}{c}{$b_6$} & \multicolumn{1}{c}{$b_{...}$} \\
                    \hhline{~*{10}{-}}
            $v_1$ &\cellcolor[HTML]{C1C1C1} x & \cellcolor[HTML]{C1C1C1} x & \cellcolor[HTML]{C1C1C1} & & & & & & &\\ 
            \hhline{~*{10}{-}}
            $v_2$ &\cellcolor[HTML]{C1C1C1} x & \cellcolor[HTML]{C1C1C1} x & \cellcolor[HTML]{C1C1C1} & & & & & & &\\ \hhline{~*{10}{-}}
            $v_3$ &\cellcolor[HTML]{C1C1C1} x & \cellcolor[HTML]{C1C1C1} x & \cellcolor[HTML]{C1C1C1} & & & & & & &\\ \hhline{~*{10}{-}}
            $v_4$ &\cellcolor[HTML]{C1C1C1} x & \cellcolor[HTML]{C1C1C1} x & \cellcolor[HTML]{C1C1C1} & & & & & & &\\ \hhline{~*{10}{-}}
            $v_5$ &\cellcolor[HTML]{C1C1C1} x & \cellcolor[HTML]{C1C1C1} x & \cellcolor[HTML]{C1C1C1} & & & & & & &\\ \hhline{~*{10}{-}}
            $v_6$ &\cellcolor[HTML]{C1C1C1} x & \cellcolor[HTML]{C1C1C1} x & \cellcolor[HTML]{C1C1C1} & & & & & & &\\ \hhline{~*{10}{-}}
            $v_7$ & &  & & \cellcolor[HTML]{C1C1C1}x & \cellcolor[HTML]{C1C1C1}x & \cellcolor[HTML]{C1C1C1} &\cellcolor[HTML]{C1C1C1}  & \cellcolor[HTML]{C1C1C1}   & \cellcolor[HTML]{C1C1C1} & \cellcolor[HTML]{C1C1C1}  \\  \hhline{~*{10}{-}}
            $v_8$ & & & & \cellcolor[HTML]{C1C1C1} x & \cellcolor[HTML]{C1C1C1}x & \cellcolor[HTML]{C1C1C1} &\cellcolor[HTML]{C1C1C1}  & \cellcolor[HTML]{C1C1C1}   & \cellcolor[HTML]{C1C1C1} & \cellcolor[HTML]{C1C1C1} \\  \hhline{~*{10}{-}}
            $v_9$ & & & & \cellcolor[HTML]{C1C1C1}& \cellcolor[HTML]{C1C1C1}& \cellcolor[HTML]{C1C1C1}x &\cellcolor[HTML]{C1C1C1} x & \cellcolor[HTML]{C1C1C1}   & \cellcolor[HTML]{C1C1C1} & \cellcolor[HTML]{C1C1C1} \\  \hhline{~*{10}{-}}
            $v_{10}$ & & & & \cellcolor[HTML]{C1C1C1}& \cellcolor[HTML]{C1C1C1}& \cellcolor[HTML]{C1C1C1}x &\cellcolor[HTML]{C1C1C1} x & \cellcolor[HTML]{C1C1C1} & \cellcolor[HTML]{C1C1C1}& \cellcolor[HTML]{C1C1C1} \\  \hhline{~*{10}{-}}
            $v_{11}$ & &  & & \cellcolor[HTML]{C1C1C1}& \cellcolor[HTML]{C1C1C1}& \cellcolor[HTML]{C1C1C1} &\cellcolor[HTML]{C1C1C1}  & \cellcolor[HTML]{C1C1C1} x  & \cellcolor[HTML]{C1C1C1} x & \cellcolor[HTML]{C1C1C1}\\  \hhline{~*{10}{-}}
            $v_{12}$ & &  & & \cellcolor[HTML]{C1C1C1}& \cellcolor[HTML]{C1C1C1}& \cellcolor[HTML]{C1C1C1} &\cellcolor[HTML]{C1C1C1}  & \cellcolor[HTML]{C1C1C1} x  & \cellcolor[HTML]{C1C1C1} x & \cellcolor[HTML]{C1C1C1}\\  \hline \hline
            LV & \cellcolor[HTML]{1d53ab} &  \cellcolor[HTML]{1d53ab} &  & \cellcolor[HTML]{1d53ab}  &\cellcolor[HTML]{1d53ab} &  \cellcolor[HTML]{1d53ab} & \cellcolor[HTML]{1d53ab} & & &\\ \hhline{~*{10}{-}}
        \end{tabular}
    \end{table} 
\end{example}

\subsection{CC- and PAV-improvements}

\begin{example}\label{ex:neg-CCgain}
The CC-improvement of LV can be below 1 in a party-list profile without a broadcasting order.
Consider an election with $k=4, l=2$ and a party-list profile where the largest party has 5 voters and the second party has 4 voters. If the 5 voters in $P_1$ all spread their votes and all vote for two different candidates, while the 4 voters in $P_2$ pair up and two of them vote for candidates $a$ and $b$, while two others vote for $c$ and $d$, then $a,b,c$ and $d$ will get 2 votes each, while all candidates of party $P_1$ only get 1 vote. Hence, the winning committee of LV $\{a,b,c,d\}$ has a CC-score of $|P_2|=4$, which is smaller than $|P_1|=5$, which is the CC-score of the winning committee of AV: 
the CC-improvement is $\frac{4}{5}$.
\end{example}

\begin{example}[PAV-improvement in general elections]\label{ex:PAV_general}
We start by noting that the PAV-
improvement is not necessarily greater than 1.
Consider the situation where the limited votes are such that less than $k$ candidates get any vote from the limited ballots. Then the rest of $k$ has to be filled up with a tie-breaking rule, and it is therefore easy to construct an election where the PAV-score of AV is higher than that of LV. 
Consider for instance the very simple case where $l=3, k=4$, we have 6 candidates $c_1, ..., c_6$ and only one voter $i$, with $A_i = \{c_1, c_2, c_3, c_5, c_6\}$. Suppose the voter submits the first three of his approved candidates as his limited ballot, and that we use lexicographical tie-breaking for the remaining candidates. Then $\{c_1, c_2, c_3, c_4\}$ is elected by LV with a PAV-score of $1+\frac{1}{2}+\frac{1}{3}$, while any committee that AV outputs consists only of candidates that $i$ likes,  with PAV-score of $1+\frac{1}{2}+\frac{1}{3} + \frac{1}{4}$.

Let us turn to the slightly more interesting cases where no such extreme tie-breaking is needed. 
Assume that there are at least $k$ candidates that get at least one vote. 
Despite this constraint, it is still possible for the PAV-
improvement to be below 1.
Consider for example the profile in Table \ref{tab:AV>LV}, where $l=3$, $k=4$. 
The outcome of AV is $W_{AV} = \{c_2, c_3, c_4, c_5\}$ with $s_{PAV}(A,W_{AV})=2(1+\frac{1}{2}+\frac{1}{3}+\frac{1}{4}) + 2(1+\frac{1}{2}+\frac{1}{3})\approx 7.83$, while the outcome of LV is $W_{LV}=\{c_1, c_2, c_3, c_4\}$, with $s_{PAV}(A,W_{LV}) = 2(1+\frac{1}{2}+\frac{1}{3}+\frac{1}{4}) + (1+\frac{1}{2}+\frac{1}{3}) + (1+\frac{1}{2})=7.5$. Hence, $s_{PAV}(A,W_{AV})>s_{PAV}(A,W_{LV})$. 

Importantly, the PAV-
improvement may also be above 1.
Table \ref{tab:LV>AV} provides an example where the approval ballots are more aligned than in the previous example (note that $|A_{v_4}|<l$, so $v_4$ has to vote for one candidate she does not approve of). Here, the outcome of AV is $W_{AV} = \{c_1, c_2, c_3, c_4\}$ with $s_{PAV}(A,W_{AV})=3(1+\frac{1}{2}+\frac{1}{3}+\frac{1}{4}) +0 = 6.25$, while the outcome of LV is $W_{LV}=\{ c_3, c_4, c_5, c_6\}$, with $s_{PAV}(A,W_{LV}) = 2(1+\frac{1}{2}+\frac{1}{3}) + 2(1+\frac{1}{2}) \approx 6.67$.  Hence, $s_{PAV}(A,W_{LV})>s_{PAV}(A,W_{AV})$. Note that if we do not consider broadcasted elections, there can be a lot of variation in the outcome of LV depending on which of their approved candidates the voters choose to vote on. In this example however, the PAV-score of an LV-outcome can never be lower than that of the AV-outcome, since that has the lowest possible PAV-score for a 4-candidate committee.
\end{example}
\begin{table*}[t]
    \caption{Two example elections with $l=3$, $k=4$. Approvals in grey, votes indicated by x, the bottom rows show the winning committees according to LV (blue) and AV (red), the columns on the right indicate the number of approved candidates (representatives) of a voter in the winning committees.
    \vspace{15pt}
    }
    \small
    \setlength{\tabcolsep}{3pt}
    \renewcommand{\arraystretch}{0.8}
    \begin{subtable}[h]{0.5\textwidth}
        \caption{
        $s_{PAV}(A,W_{AV})\approx 7.83$
        $>s_{PAV}(A,W_{LV})= 7.5$.}
        \label{tab:AV>LV}
        \centering
        \captionsetup{width=.9\linewidth}
        \begin{tabular}{c|c|c|c|c|c|c||c|c|}
        \multicolumn{1}{c}{} & \multicolumn{1}{c}{} & \multicolumn{1}{c}{} & \multicolumn{1}{c}{} & \multicolumn{1}{c}{} & \multicolumn{1}{c}{} & \multicolumn{1}{c}{} & \multicolumn{2}{c}{repr.} \\
                   \multicolumn{1}{c}{} & \multicolumn{1}{c}{$c_1$} & \multicolumn{1}{c}{$c_2$} & \multicolumn{1}{c}{$c_3$} & \multicolumn{1}{c}{$c_4$} & \multicolumn{1}{c}{$c_5$} & \multicolumn{1}{c}{$c_6$} & \multicolumn{1}{c}{LV} & \multicolumn{1}{c}{AV} \\
                    \hhline{~*{8}{-}} 
            $v_1$ & \cellcolor[HTML]{C1C1C1} x & \cellcolor[HTML]{C1C1C1} x & \cellcolor[HTML]{C1C1C1} x & \cellcolor[HTML]{C1C1C1} & \cellcolor[HTML]{C1C1C1} & & 4 & 4\\ \hhline{~*{8}{-}} 
            $v_2$ & & \cellcolor[HTML]{C1C1C1} x & \cellcolor[HTML]{C1C1C1} x & \cellcolor[HTML]{C1C1C1} x & \cellcolor[HTML]{C1C1C1} & & 3 & 4 \\ \hhline{~*{8}{-}} 
            $v_3$ & &  & \cellcolor[HTML]{C1C1C1} x & \cellcolor[HTML]{C1C1C1} x & \cellcolor[HTML]{C1C1C1} x &\cellcolor[HTML]{C1C1C1} & 2 & 3\\ \hhline{~*{8}{-}} 
            $v_4$ & \cellcolor[HTML]{C1C1C1} x & \cellcolor[HTML]{C1C1C1} x & \cellcolor[HTML]{C1C1C1} x & \cellcolor[HTML]{C1C1C1} &  & & 4 & 3 \\ \hline
            \multicolumn{1}{c}{} & \multicolumn{6}{c}{winning committees}& \multicolumn{2}{c}{PAV-score} \\   
            \hline
            LV & \cellcolor[HTML]{1d53ab}& \cellcolor[HTML]{1d53ab} &\cellcolor[HTML]{1d53ab} & \cellcolor[HTML]{1d53ab} &  &  & 7.5 & \\ \hhline{~*{8}{-}} 
            AV & & \cellcolor[HTML]{d64238}& \cellcolor[HTML]{d64238} &\cellcolor[HTML]{d64238} & \cellcolor[HTML]{d64238} & & & 7.83 \\ \hhline{~*{8}{-}} 
        \end{tabular}
    \end{subtable}%
    \hfill
    \begin{subtable}[h]{0.5\textwidth}
    \caption{
    $s_{PAV}(A,W_{AV})= 6.25$
    $<s_{PAV}(A,W_{LV}) \approx 6.67$.}
    \label{tab:LV>AV}
      \centering
      \small
      \captionsetup{width=.9\linewidth}
        \begin{tabular}{c|c|c|c|c|c|c||c|c|}
          \multicolumn{1}{c}{} & \multicolumn{1}{c}{} & \multicolumn{1}{c}{} & \multicolumn{1}{c}{} & \multicolumn{1}{c}{} & \multicolumn{1}{c}{} & \multicolumn{1}{c}{} & \multicolumn{2}{c}{repr.} \\
                     \multicolumn{1}{c}{} & \multicolumn{1}{c}{
                     $c_1$} & \multicolumn{1}{c}{$c_2$} & \multicolumn{1}{c}{
                     $c_3$} & \multicolumn{1}{c}{$c_4$} & \multicolumn{1}{c}{$c_5$} & \multicolumn{1}{c}{$c_6$} &  \multicolumn{1}{c}{LV} &  \multicolumn{1}{c}{AV}  \\
                    \hhline{~*{8}{-}} 
            $v_1$ & \cellcolor[HTML]{C1C1C1} x & \cellcolor[HTML]{C1C1C1} x & \cellcolor[HTML]{C1C1C1} x & \cellcolor[HTML]{C1C1C1} &  & & 2 & 4 \\ \hhline{~*{8}{-}} 
            $v_2$ & \cellcolor[HTML]{C1C1C1} & \cellcolor[HTML]{C1C1C1} & \cellcolor[HTML]{C1C1C1} x & \cellcolor[HTML]{C1C1C1} x &  & \cellcolor[HTML]{C1C1C1} x & 3 & 4\\ \hhline{~*{8}{-}} 
            $v_3$ & \cellcolor[HTML]{C1C1C1} & \cellcolor[HTML]{C1C1C1} & \cellcolor[HTML]{C1C1C1} x & \cellcolor[HTML]{C1C1C1} x & \cellcolor[HTML]{C1C1C1} x & & 3 & 4\\ \hhline{~*{8}{-}} 
            $v_4$ & & & x &  &  \cellcolor[HTML]{C1C1C1} x &\cellcolor[HTML]{C1C1C1} x & 2 & 0 \\ \hline 
                    \multicolumn{1}{c}{} & \multicolumn{6}{c}{winning committees}& \multicolumn{2}{c}{PAV-score} \\  \hline
            LV & & & \cellcolor[HTML]{1d53ab}& \cellcolor[HTML]{1d53ab} &\cellcolor[HTML]{1d53ab} & \cellcolor[HTML]{1d53ab} & 6.67 & \\ \hhline{~*{8}{-}} 
        AV & \cellcolor[HTML]{d64238}& \cellcolor[HTML]{d64238} &\cellcolor[HTML]{d64238} & \cellcolor[HTML]{d64238} & & & & 6.25\\ \hhline{~*{8}{-}}
        \end{tabular}
    \end{subtable} 
\end{table*}

\subsection{LPAV and LSAV}

\paragraph{Limited proportional approval voting}
On broadcasted party-list elections, limited PAV (LPAV) can not decrease the CC-score compared to PAV since it can only `cut' candidates (but leave at least $l$) from parties that deserve more according to PAV and give them to parties that deserve less, so there cannot be parties that did get a candidate elected in PAV and do not get anything in LPAV. If it is proportional to leave some voters unrepresented, limiting the ballots can (at the cost of proportionality) increase diversity. 
\begin{example}[LPAV increases diversity]\label{ex:LPAV-PAV}
Consider a broadcasted party-list election with 5 voters, where $v_1, ..., v_4$ approve $P_1$ and $v_5$ approves $P_2$. If $k=3$, PAV will elect only candidates from $P_1$ since $P_2$ is too small to `deserve' any. However, with $l=2$, LPAV will elect two candidates from $P_1$ and one candidate from $P_2$, so also $v_5$ is represented. 
\end{example}

On non-party-list elections however, this is not necessarily the case:
\begin{example}[LPAV decreases diversity]
    We construct an election where the outcome of PAV has a higher CC-score than the outcome of LPAV: limiting the ballot size decreases the diversity in this case. See the election in Table \ref{tab:LPAV-PAV} where $s_{CC}(PAV)>s_{CC}(LPAV)$.
    \begin{table}[t]
    \caption{$l=3$, $k=4$. Approvals in grey, limited votes indicated by x. PAV selects $\{c_2, c_3, c_4, c_5\}$ (red), LPAV selects $\{c_1, c_2, c_3, c_4\}$ (blue). The columns on the right indicate whether a voter is represented by the winning committees.}
        \label{tab:LPAV-PAV}
    \small
    \centering
    \renewcommand{\arraystretch}{0.65}
    \setlength{\tabcolsep}{4pt}
        \begin{tabular}{c|c|c|c|c|c|c|c|c|| c|c|}
         \multicolumn{1}{c}{}& \multicolumn{1}{c}{}& \multicolumn{1}{c}{}& \multicolumn{1}{c}{}& \multicolumn{1}{c}{}& \multicolumn{1}{c}{}& \multicolumn{1}{c}{}& \multicolumn{1}{c}{}& \multicolumn{1}{c}{}&\multicolumn{2}{c}{repr.} \\
                    \multicolumn{1}{c}{}& \multicolumn{1}{c}{
                    $c_1$} & \multicolumn{1}{c}{
                    $c_2$} & \multicolumn{1}{c}{$c_3$} & \multicolumn{1}{c}{$c_4$} & \multicolumn{1}{c}{$c_5$} & \multicolumn{1}{c}{$c_6$} & \multicolumn{1}{c}{$c_7$} & \multicolumn{1}{c}{$c_8$} &\multicolumn{1}{c}{LPAV} &\multicolumn{1}{c}{PAV}   \\
                    \hhline{~*{10}{-}}
            $v_1$ & \cellcolor[HTML]{C1C1C1} x & \cellcolor[HTML]{C1C1C1} x & \cellcolor[HTML]{C1C1C1} x & \cellcolor[HTML]{C1C1C1} & \cellcolor[HTML]{C1C1C1} &  & & & \cellcolor[HTML]{1d53ab}& \cellcolor[HTML]{d64238}\\ 
            \hhline{~*{10}{-}}
            $v_2$ & \cellcolor[HTML]{C1C1C1} x & \cellcolor[HTML]{C1C1C1}  & \cellcolor[HTML]{C1C1C1}  & \cellcolor[HTML]{C1C1C1} x & \cellcolor[HTML]{C1C1C1}x & & & & \cellcolor[HTML]{1d53ab}& \cellcolor[HTML]{d64238}\\ 
            \hhline{~*{10}{-}}
            $v_3$ & &  & \cellcolor[HTML]{C1C1C1} x & \cellcolor[HTML]{C1C1C1} x & \cellcolor[HTML]{C1C1C1} x &\cellcolor[HTML]{C1C1C1} & & & \cellcolor[HTML]{1d53ab}& \cellcolor[HTML]{d64238}\\ 
            \hhline{~*{10}{-}}
            $v_4$ & \cellcolor[HTML]{C1C1C1} x & \cellcolor[HTML]{C1C1C1} x & \cellcolor[HTML]{C1C1C1} x & \cellcolor[HTML]{C1C1C1} &  & & & & \cellcolor[HTML]{1d53ab}& \cellcolor[HTML]{d64238}\\ 
            \hhline{~*{10}{-}}
            $v_5$ & & \cellcolor[HTML]{C1C1C1} x & \cellcolor[HTML]{C1C1C1} x & \cellcolor[HTML]{C1C1C1} x &  & & & & \cellcolor[HTML]{1d53ab}& \cellcolor[HTML]{d64238}\\ 
            \hhline{~*{10}{-}}
            $v_6$ & &  & &  & \cellcolor[HTML]{C1C1C1} &\cellcolor[HTML]{C1C1C1} x & \cellcolor[HTML]{C1C1C1} x  & \cellcolor[HTML]{C1C1C1} x & & \cellcolor[HTML]{d64238}\\  \hline 
                    \multicolumn{1}{c}{} & \multicolumn{8}{c}{winning committees}& \multicolumn{2}{c}{CC-score} \\ \hline
            LPAV & \cellcolor[HTML]{1d53ab} & \cellcolor[HTML]{1d53ab} & \cellcolor[HTML]{1d53ab}  &\cellcolor[HTML]{1d53ab} &  & & & & 5& \\ 
            \hhline{~*{10}{-}}
            PAV & &\cellcolor[HTML]{d64238} & \cellcolor[HTML]{d64238} & \cellcolor[HTML]{d64238}  &\cellcolor[HTML]{d64238} &  & & & & 6\\ 
            \hhline{~*{10}{-}}
        \end{tabular}
    \end{table} 
\end{example}

\paragraph{Limited satisfaction approval voting}
In party-list elections, SAV is not a natural voting rule to choose, since there is no obvious reason to penalize parties for their size (number of candidates). It can give very non-diverse outcomes, since small parties are prioritized. In such settings, limiting the ballot size (LSAV) can help mitigate this effect.
\begin{example}\label{ex:LSAV}
 Take for example a broadcasted party-list election with eleven voters, where the first ten voters vote for party $P_1$ that has 100 candidates, and the eleventh voter votes for $P_2$ with only five candidates. Suppose that $k=4$. Then SAV will give 0.1 point to every candidate in $P_1$ and 0.2 point to every candidate in $P_2$, so elect four candidates from $P_2$, which only represent $v_{11}$. LSAV with $l=2$, however, will limit the votes to two per voter, and will therefore give 5 points to the first two candidates in $P_1$ and 0.5 point to the first two candidates in $P_2$, and therefore select four candidates from $P_1$, which gives a CC-score of 10 instead of 1.
\end{example}
However, there are also broadcasted party-list elections where the CC-score of SAV is higher than that of LSAV. When there are many small parties with only a few voters and a few candidates, and one party with a bit more voters and a lot more candidates, LSAV will elect candidates from only the larger party (because the candidates there have more supporters), while SAV will select candidates from all small parties.

\subsection{Examples of LV-games}

\begin{example}
    Consider the profile in Example \ref{ex:non-strategic}, as displayed in Table \ref{tab:example_non_strategic}. Here, if both party $P_a$ and party $P_b$ would play a lower quota strategy, they would both choose $\lfloor\frac{n_a\cdot k}{|N|}\rfloor = \lfloor\frac{n_b\cdot k}{|N|}\rfloor = \lfloor\frac{6\cdot 6}{12}\rfloor = 3$ candidates, that would each get 4 votes. For example, $s_a = ((\{a_1, a_2\},2), (\{a_2, a_3\},2), (\{a_1, a_3\},2))$ and $s_b = ((\{b_1, b_2\},2), (\{b_2, b_3\},2), (\{b_1, b_3\},2))$. Then indeed the winning committee would be $\{a_1, a_2, a_3, b_1, b_2, b_3\}$, so both parties would get their lower quota  of seats. 
    In this example, $\sum_{i\in\{a,b\}}\lfloor\frac{k\cdot n_i}{|N|}\rfloor = k$, so according to Proposition \ref{prop:NE_lower_quota} the strategy profile should be a Nash equilibrium. This is indeed the case: neither party $P_a$ nor party $P_b$ can divide the votes of their voters over candidates such that more than 3 candidates get at least 4 votes, so neither party can steal any seat from the other party by deviating from their current strategy.
\end{example}

\begin{example}\label{ex:lq}
    Consider a party-list election with $l=2$, $k=10$ and three parties, $P_a$, $P_b$, and $P_c$, as shown in Table \ref{tab:example_lq}. Here, if all parties would play a lower quota strategy, party $P_a$ would choose $\lfloor\frac{n_a\cdot k}{|N|}\rfloor =4$ candidates, parties $P_b$ and $P_c$ would choose $\lfloor\frac{n_b\cdot k}{|N|}\rfloor = \lfloor\frac{n_c\cdot k}{|N|}\rfloor=2$ candidates. In party $P_a$, the 4 chosen candidates would get $\frac{3\cdot2}{4}=1.5$ votes on average, so two candidates would get one vote and two candidates would get two votes, in parties $P_b$ and $P_c$, the two chosen candidates would get $\frac{2\cdot2}{2}=2$ votes each. For example, $s_a = ((\{a_1, a_2\},1), (\{a_3, a_4\},2))$, $s_b = ((\{b_1, b_2\},2))$, and $s_c = ((\{c_1, c_2\},2))$, as in Table \ref{tab:example_lq}.
    With these strategies, all winning committees will contain $\{a_1, a_2, a_3, a_4, b_1, b_2, c_1, c_2\}$, together with two other candidates that did not get any votes. Every party gets their lower quota, but the profile is not a Nash equilibrium: for example party $P_a$ can `claim' the two remaining seats by changing its strategy to $s_a = ((\{a_1, a_2\},1), (\{a_3, a_4\},1) (\{a_5, a_6\},1))$. However, it is a 2-Nash equilibrium: no party can gain more than 2 assured seats by deviating from their current strategy.
    \begin{table}[ht]
    \caption{Party-list profile used in Example \ref{ex:lq}. $l=2$, $k=10$. Approvals in grey, limited votes indicated by x.}
    \label{tab:example_lq}
    \centering
    \setlength{\tabcolsep}{3pt}
    \renewcommand{\arraystretch}{0.8}
    \begin{tabular}{c|c|c|c|c|c|c|c|c|c|c|c|c|c|c|}
        \multicolumn{1}{c}{}& \multicolumn{1}{c}{$a_1$} & \multicolumn{1}{c}{$a_2$} & \multicolumn{1}{c}{$a_3$} & \multicolumn{1}{c}{$a_4$} & \multicolumn{1}{c}{$a_5$} & \multicolumn{1}{c}{$a_5$}& \multicolumn{1}{c}{$b_1$} & \multicolumn{1}{c}{$b_2$} & \multicolumn{1}{c}{$b_3$} & \multicolumn{1}{c}{$b_3$} & \multicolumn{1}{c}{$c_1$} & \multicolumn{1}{c}{$c_2$} & \multicolumn{1}{c}{$c_3$} &\multicolumn{1}{c}{$c_3$}\\
        \hhline{~*{14}{-}} 
        $v_1$ & \cellcolor[HTML]{C1C1C1}x & \cellcolor[HTML]{C1C1C1}x & \cellcolor[HTML]{C1C1C1} & \cellcolor[HTML]{C1C1C1} & \cellcolor[HTML]{C1C1C1}  &\cellcolor[HTML]{C1C1C1}  & & & & & & & &\\ \hhline{~*{14}{-}} 
        $v_2$ & \cellcolor[HTML]{C1C1C1}  & \cellcolor[HTML]{C1C1C1}  & \cellcolor[HTML]{C1C1C1}x & \cellcolor[HTML]{C1C1C1}x & \cellcolor[HTML]{C1C1C1} & \cellcolor[HTML]{C1C1C1}  & & & & & & & &\\ \hhline{~*{14}{-}} 
        $v_3$ & \cellcolor[HTML]{C1C1C1}  & \cellcolor[HTML]{C1C1C1}  & \cellcolor[HTML]{C1C1C1}x & \cellcolor[HTML]{C1C1C1}x& \cellcolor[HTML]{C1C1C1} & \cellcolor[HTML]{C1C1C1}  & & & & & && &\\ \hhline{~*{14}{-}} 
        $v_4$ &  & & & & & & \cellcolor[HTML]{C1C1C1}x &\cellcolor[HTML]{C1C1C1}x   &\cellcolor[HTML]{C1C1C1}  & \cellcolor[HTML]{C1C1C1}  & & &  &\\ \hhline{~*{14}{-}} 
        $v_5$ &  & & & & & & \cellcolor[HTML]{C1C1C1}x &\cellcolor[HTML]{C1C1C1}x   &\cellcolor[HTML]{C1C1C1}  & \cellcolor[HTML]{C1C1C1}  & & & & \\ \hhline{~*{14}{-}} 
        $v_6$ &  & & & & & & & & & & \cellcolor[HTML]{C1C1C1}x &\cellcolor[HTML]{C1C1C1}x   &\cellcolor[HTML]{C1C1C1} &\cellcolor[HTML]{C1C1C1}   \\ \hhline{~*{14}{-}} 
        $v_7$ &  & & & & & & & & & & \cellcolor[HTML]{C1C1C1}x &\cellcolor[HTML]{C1C1C1}x   &\cellcolor[HTML]{C1C1C1} & \cellcolor[HTML]{C1C1C1}   \\ \hhline{~*{14}{-}} 
    \end{tabular}
\end{table} 
\end{example}

\section{Omitted proofs}\label{app:proofs}
\subsection{Proof of Theorem \ref{thm:CC-guarantee}.}

\noindent\textbf{Theorem \ref{thm:CC-guarantee}.}    \emph{The CC-guarantee of LV is 0.}

\begin{proof}
    The idea of the proof is that a small group that coordinates well can overrule a large group that coordinates poorly. Let $\mathcal{E}$ be a class of elections of the following form: for any $E\in \mathcal{E}$ we have $E= ( N, C, k, l, A, L)$ where $k\geq 2$ is divisible by $l$\footnote{This assumption is not necessary for the proof to work, but makes it easier to read.}, $N$ consists of two groups $N=X\cup S$ with $X\cap S = \emptyset$, and $|S|=2\frac{k}{l}$; for any $j\in S$, $A_j= W$ where $W\subseteq C$ with $|W|=k$ and for any $j\in X$, $A_j = Y$, where $Y\subseteq C$ with $|Y|\geq l|X|$ and $Y\cap W = \emptyset$; for any $j, j' \in X$, $L_j\cap L_{j'}=\emptyset$ (so all candidates in $Y$ get at most one vote), for any $j \in S$ there is exactly one $j'\in S$ such that $L_j = L_{j'}$, and for all other $j''$, $L_j\cap L_{j''}=\emptyset$ (so all candidates in $W$ get two votes). In any such election, $W$ wins according to LV, and $s_{CC}(W) = |S|=2\frac{k}{l}$. The most diverse committee $W'$ is one that contains at least one candidate from $W$ and at least one candidate from $Y$, with $s_{CC}(W')=2\frac{k}{l} + |X|$. This gives us that for $E\in \mathcal{E}$, as $|X|$ becomes larger, $\frac{\min_{W\in LV(E)}s_{CC}(A,W)}{\max_{W\in S_k(C)}s_{CC}(A, W)} = \frac{2\frac{k}{l}}{2\frac{k}{l} + |X|}$ approximates 0, and therefore, the CC-guarantee of LV is $\kappa_{cc}(k) = \inf_{A\in \mathcal{A}}\frac{\min_{W\in LV(E)}s_{CC}(A,W)}{\max_{W\in S_k(C)}s_{CC}(A,W)}=0$.
\end{proof}


\subsection{LV fails the main proportionality axioms}

We briefly recall the definitions of some essential proportionality axioms:
 A voting rule $\mathcal{R}$ satisfies \emph{justified representation (JR)} \citep{Aziz2017} if for each election $E = ( N, C, k, l, A, L)$, for each $W\in \mathcal{R}(E)$ and each 1-cohesive group of voters $S$, there is a voter $i\in S$ who is represented by at least one member of $W$, i.e. $|W\cap A_i|\geq 1$, where a group $S$ is called 1-cohesive if $|S|\geq \frac{n}{k}$ and $|\cap_{i\in S}A_i|\geq 1$.  The definition is extended to capture the idea that groups that agree on more candidates should be represented by more candidates: A rule satisfies \emph{extended justified representation (EJR)} \citep{Aziz2017} if for each election $E = ( N, C, k, l, A, L)$, for each $W\in \mathcal{R}(E)$ and each $\ell$-cohesive group of voters $S$ (for $\ell\leq k$), there is a voter $i\in S$ who is represented by at least $\ell$ members of $W$, i.e. $|W\cap A_i|\geq \ell$, where a group $S$ is $\ell$-cohesive if $|S|\geq \ell\cdot \frac{n}{k}$ and $|\cap_{i\in S}A_i|\geq \ell$.
 The notion of JR is generalized to a weaker condition than EJR, \emph{proportional justified representation (PJR)} \citep{sanchezfernandez2016}: A voting rule $\mathcal{R}$ satisfies PJR if for each election $E = ( N, C, k, l, A, L)$, for each $W\in \mathcal{R}(E)$ and each $\ell$-cohesive group of voters $S$, the collective group (instead of just one member of the group) has at least $\ell$ representatives: $|W\cap (\cup_{i\in S}A_i)|\geq \ell$.
 A rule $\mathcal{R}$ satisfies \emph{laminar proportionality} \citep{Peters2020limwelf} if, for every laminar election $E = ( N, C, k, l, A, L)$\footnote{See Appendix \ref{app:laminar} for a definition of broadcasted laminar elections.} it returns a committee that satisfies the following: (1) if $E$ is unanimous, then $\mathcal{R}(E)\subseteq C$; (2) if there is a candidate $c\in C$ such that $c\in A_i$ for all $i \in N$ and the election $E_{-c}$ is laminar, then $\mathcal{R}(E)=W'\cup\{c\}$, where $W'$ is laminar proportional for $E_{-c}$; and (3) if $E$ is the sum of $E_1$ and $E_2$, then $\mathcal{R}(E)=W_1\cup W_2$, where $W_1$ is laminar proportional for $E_1$ and $W_2$ is laminar proportional for $E_2$.
A rule $\mathcal{R}$ satisfies \emph{priceability} \citep{Peters2020limwelf} if for any election $E = ( N, C, k, l, A, L)$, $\mathcal{R}(E)$ is priceable: there exists a price system $\textbf{ps}=(p,\{p_i\}_{i\in N}$ that supports it, where $p>0$ is a price and $p_i$ are payment functions $p_i:C\rightarrow [0,1]$ such that: (1) if $p_i(c)>0$, then $c\in A_i$; (2) $\sum_{c\in C}p_i(c)\leq 1$; (3) for all $c\in \mathcal{R}(E)$, $\sum_{i\in N}p_i(c)=p$; (4) for all $c\notin \mathcal{R}(E)$, $\sum_{i\in N}p_i(c)=0$; and (5) for all $c\notin \mathcal{R}(E)$: $\sum_{i\in N:c\in A_i}\left(1-\sum_{c\in \mathcal{R}(E)}p_i(c')\right)\leq p$.

A rule $\mathcal{R}$ satisfies \emph{lower quota} if for each party-list election $E$ with parties $P_1, ..., P_g$, if every party $P_i$ gets at least  $\lfloor\frac{k\cdot n_i}{|N|}\rfloor$ seats \citep{Lackner2023, Peters2021}. 

\medskip
\noindent\textbf{Proposition \ref{prop:axioms}.} \emph{LV fails JR, EJR, PJR, 
priceability, and lower quota even on
broadcasted party-list elections.}

\begin{proof}
    We show that LV does not satisfy JR by constructing a counterexample. Take an election with $k=4$ and $l=3$, where we have 8 voters that vote as shown in Table \ref{tab:JR-counter}. 
    \begin{table}[b]
    \caption{Example used in the proof of Proposition \ref{prop:axioms}. $l=3$, $k=4$. Approvals in grey, limited votes indicated by x.}
    \label{tab:JR-counter}
    \centering
    \setlength{\tabcolsep}{3pt}
    \renewcommand{\arraystretch}{0.8}
    \begin{tabular}{c|c|c|c|c|c|c|c|c|c|}
        \multicolumn{1}{c}{}& \multicolumn{1}{c}{$a_1$} & \multicolumn{1}{c}{$a_2$} & \multicolumn{1}{c}{$a_3$} & \multicolumn{1}{c}{$w_1$} & \multicolumn{1}{c}{$w_2$} & \multicolumn{1}{c}{$w_3$} & \multicolumn{1}{c}{$w_4$} & \multicolumn{1}{c}{$w_5$} & \multicolumn{1}{c}{$w_6$} \\
                \hhline{~*{9}{-}} 
        $v_1$ & \cellcolor[HTML]{C1C1C1} x & \cellcolor[HTML]{C1C1C1} x & \cellcolor[HTML]{C1C1C1} x &  &  & & & & \\ \hhline{~*{9}{-}} 
        $v_2$ & \cellcolor[HTML]{C1C1C1} x & \cellcolor[HTML]{C1C1C1} x & \cellcolor[HTML]{C1C1C1} x &  &  &  & & &\\ \hhline{~*{9}{-}} 
        $v_3$ &  & &  & \cellcolor[HTML]{C1C1C1} x &\cellcolor[HTML]{C1C1C1} x   &\cellcolor[HTML]{C1C1C1} x & & &  \\ \hhline{~*{9}{-}} 
        $v_4$ &  & &  & \cellcolor[HTML]{C1C1C1} x &\cellcolor[HTML]{C1C1C1} x   &\cellcolor[HTML]{C1C1C1} x & & &  \\ \hhline{~*{9}{-}} 
        $v_5$ &  & &  & \cellcolor[HTML]{C1C1C1} x &\cellcolor[HTML]{C1C1C1} x   &\cellcolor[HTML]{C1C1C1} x & & &  \\ \hhline{~*{9}{-}} 
        $v_6$ &  & &  & & &  & \cellcolor[HTML]{C1C1C1} x &\cellcolor[HTML]{C1C1C1} x   &\cellcolor[HTML]{C1C1C1} x \\ \hhline{~*{9}{-}} 
        $v_7$ &  & &  & & &  & \cellcolor[HTML]{C1C1C1} x &\cellcolor[HTML]{C1C1C1} x   &\cellcolor[HTML]{C1C1C1} x \\ \hhline{~*{9}{-}} 
        $v_8$ &  & &  & & &  & \cellcolor[HTML]{C1C1C1} x &\cellcolor[HTML]{C1C1C1} x   &\cellcolor[HTML]{C1C1C1} x \\ \hhline{~*{9}{-}} 
    \end{tabular}
\end{table} 
The group $V = \{v_1, v_2\}$ is 1-cohesive, since $|V|=2\geq \frac{n}{k}=\frac{8}{4}$, $v_1$ and $v_2$ agree on all their votes. However, the candidates they vote for all get two votes, while all other candidates $w_1, ..., w_6$ get three votes. Hence, there is no voter in $V$ who is represented by at least one member in the winning committee. Note that in this example limited voting and approval voting are equivalent because all voters approve exactly 3 candidates, this is however not crucial for the counterexample to work. Since both EJR and priceability imply PJR and PJR implies JR \citep{Peters2020limwelf, sanchezfernandez2016}, and since lower quota implies JR on party-list profiles \citep{Peters2021} (and the counterexample is indeed a party-list profile), we can conclude that LV does not satisfy any of the above axioms. 
\end{proof}

We conclude by also giving an example of a broadcasted laminar election where the winning committee of LV is not laminar proportional.
See the profile in Table \ref{tab:lamprop-counter}. If $k=6$ and $l=4$, all voters vote exactly for the candidates they approve. Then $c_1, c_2, c_3, c_4, c_5$ all get two votes, and the other candidates get one vote. So the winning committee consists of $c_1, ..., c_5$, and one other candidate. An example of a winning committee is shown in 
blue in Table \ref{tab:lamprop-counter}. The instance is laminar, but all outcomes are not laminar proportional, because laminar proportionality would mean that both groups $\{v_1, v_2\}$ and $\{v_3, v_4\}$ have three candidates that they approve in the winning committee, but here the first group has only two and the second group has 4.

\begin{table}[t]
    \caption{Example to show that LV does not satisfy laminar proporitonality or priceability. $k=6$ and $l=4$, the committee in blue is a winning committee according to LV.
    }
    \label{tab:lamprop-counter}
    \centering
    \small
    \setlength{\tabcolsep}{3pt}
    \renewcommand{\arraystretch}{0.8}
    \begin{tabular}{c|c|c|c|c|c|c|c|c|c|c|c|}
       \multicolumn{1}{c}{}& \multicolumn{1}{c}{
       $c_1$} & \multicolumn{1}{c}{$c_2$} & \multicolumn{1}{c}{$c_3$} & \multicolumn{1}{c}{$c_4$} & \multicolumn{1}{c}{$c_5$} & \multicolumn{1}{c}{$c_6$} & \multicolumn{1}{c}{$c_7$} & \multicolumn{1}{c}{$c_8$}  & \multicolumn{1}{c}{$c_9$}  & \multicolumn{1}{c}{$c_{10}$}  & \multicolumn{1}{c}{$c_{11}$}   \\
                \hhline{~*{11}{-}}
        $v_1$ &\cellcolor[HTML]{C1C1C1}x & &  & &  &  & & &\cellcolor[HTML]{C1C1C1} x&\cellcolor[HTML]{C1C1C1} x&\cellcolor[HTML]{C1C1C1}x  \\ \hhline{~*{11}{-}}
        $v_2$ &\cellcolor[HTML]{C1C1C1} x& &  & &  & \cellcolor[HTML]{C1C1C1}x &\cellcolor[HTML]{C1C1C1} x&\cellcolor[HTML]{C1C1C1}x & & &  \\ \hhline{~*{11}{-}}
        $v_3$ & &\cellcolor[HTML]{C1C1C1}x &\cellcolor[HTML]{C1C1C1} x &\cellcolor[HTML]{C1C1C1}x &\cellcolor[HTML]{C1C1C1}x  &  & & & & &  \\ \hhline{~*{11}{-}}
        $v_4$ & &\cellcolor[HTML]{C1C1C1}x &\cellcolor[HTML]{C1C1C1}x  &\cellcolor[HTML]{C1C1C1}x &\cellcolor[HTML]{C1C1C1}x  &  & & & & &  \\ \hline \hline
        LV & \cellcolor[HTML]{1d53ab} & \cellcolor[HTML]{1d53ab} & \cellcolor[HTML]{1d53ab} & \cellcolor[HTML]{1d53ab} & \cellcolor[HTML]{1d53ab} & \cellcolor[HTML]{1d53ab} & & & & &  \\ \hhline{~*{11}{-}}       
    \end{tabular}
\end{table} 

\subsection{LV-games}

\medskip
\noindent\textbf{Corollary \ref{cor:lq}.}
\emph{Let an LV-game be given such that for all parties $P_i$, $l\leq \lfloor\frac{k\cdot n_i}{|N|}\rfloor \leq |P_i|$. If each party $P_i$ plays a strategy in $\mathcal{LQ}_i$, then all committees in the outcome of the game satisfy lower quota.}
\begin{proof}
     The total number of candidates that get a positive number of votes (the sum of the lower quotas of all parties) is $\sum_{i\leq g}\lfloor\frac{k\cdot n_i}{|N|}\rfloor \leq k$ and no candidates outside any of the $X_i$s get any votes, so all candidates in any of the $X_i$s get elected, and the remaining seats in the winning committees are filled with all possible combinations of other candidates (who are all tied).
\end{proof}

\noindent\textbf{Theorem \ref{prop:delta_NE_lower_quota}.}
Let an LV-game be given such that for all parties $P_i$, $l\leq \lfloor\frac{k\cdot n_i}{|N|}\rfloor \leq |P_i|$. If the excess of the sum of the lower quotas of all parties is $\epsilon = k-\sum_{i\leq g}\lfloor\frac{k\cdot n_i}{|N|}\rfloor$, a strategy profile where each party $P_i$ plays a strategy in $\mathcal{LQ}_i$ is an $\epsilon$-Nash equilibrium.
\begin{proof}
    It is clear from the proof of Proposition \ref{prop:NE_lower_quota} that any party can improve their score by $\delta$ by adding the number of remaining seats to the number of candidates they assign a positive number of votes. They cannot improve their score any more than $\delta$, since by spreading their votes over more candidates, they lower the number of votes per candidate. There might still be a committee in the outcome where they get a higher score, but this is only because of tie-breaking, and since we consider pessimistic agents, this does not change anything in the utility of the party.
\end{proof}

\newpage


\section{Detailed results from experiments}\label{app:experiment_results}

\subsection{Detailed plot of CC-improvement}
Figure 
\ref{fig:CC-improvement_all_phi} shows boxplots of the CC-improvement for all parameter combinations. Since for 
$\phi \geq 0.25$ the improvements are all close to 1, they are merged in the plot. 
\begin{figure*}[t]
    \centering
    \includegraphics[width = \textwidth]{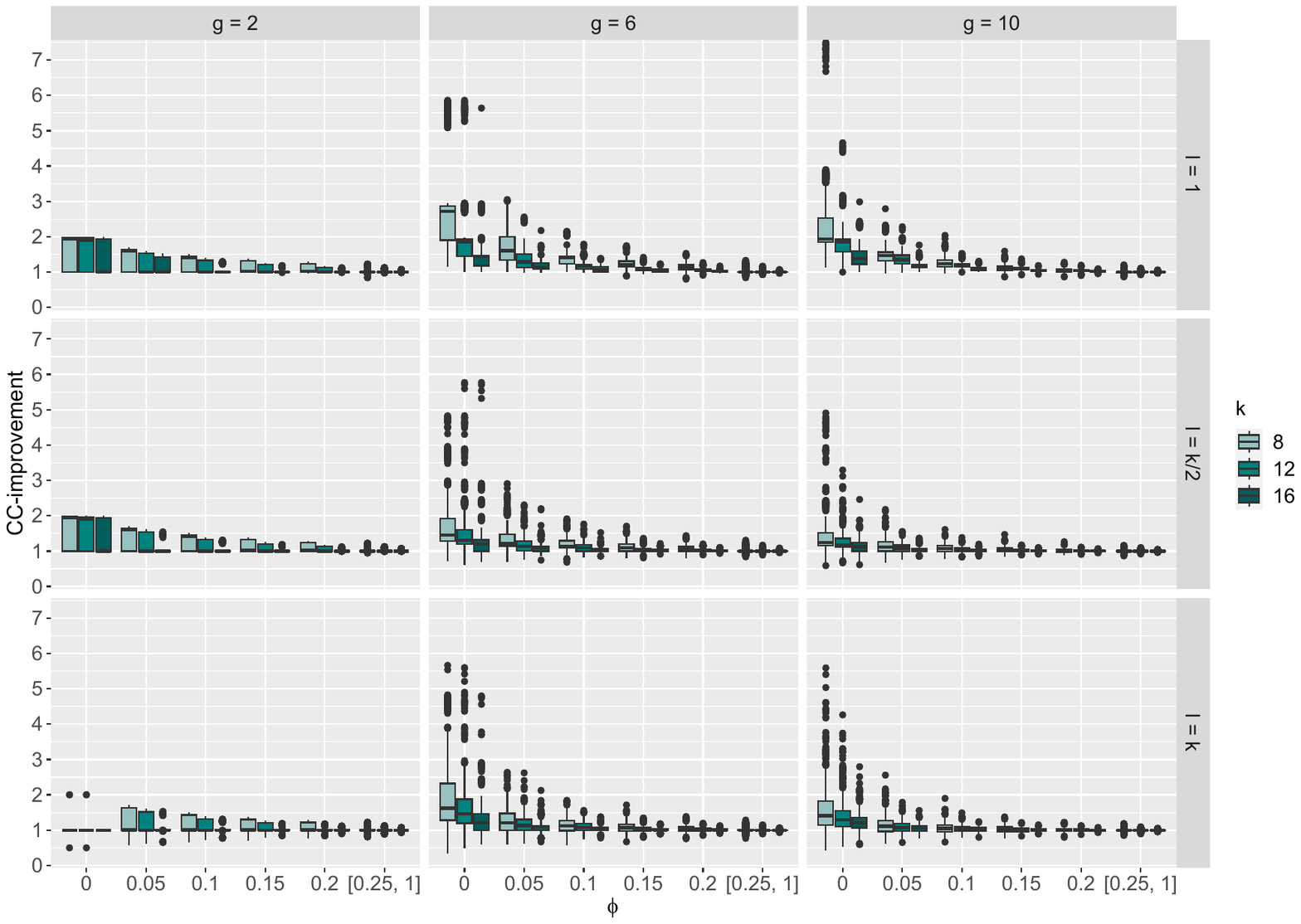}
    \caption{Boxplots of the CC-improvement for different values of $\phi$ (dispersion), $k$ (committee size), $g$ (number of parties), and $l$ (ballot limit). 
    \vspace{10pt}
    }
    \label{fig:CC-improvement_all_phi}
\end{figure*}

\subsection{PAV-improvement: effect of $\phi$}\label{app:PAV-improvement}

Figure \ref{fig:PAV-increase} shows boxplots of the PAV-improvement compared for varying values of $\phi$, $k$, $g$, and $l$.
The boxplots are rather similar to that of the CC-improvement (Figure \ref{fig:CC-improvement_all_phi}). For broadcasted party-list elections ($\phi=0$), LV has better proportionality than AV, but for $\phi\geq 0.25$, the PAV-improvement gets very close to 1 for all data points, so there is no real difference in score between AV and LV.  Just as the CC-improvement, we found that for $0< \phi < 0.25 $, the PAV-improvement decreases as $\phi$ increases, which shows that also our theoretical results on proportionality are robust: the closer an election is to a party-list election, the better LV performs compared to AV.
\begin{figure*}[t]
    \centering
    \includegraphics[width=\textwidth]{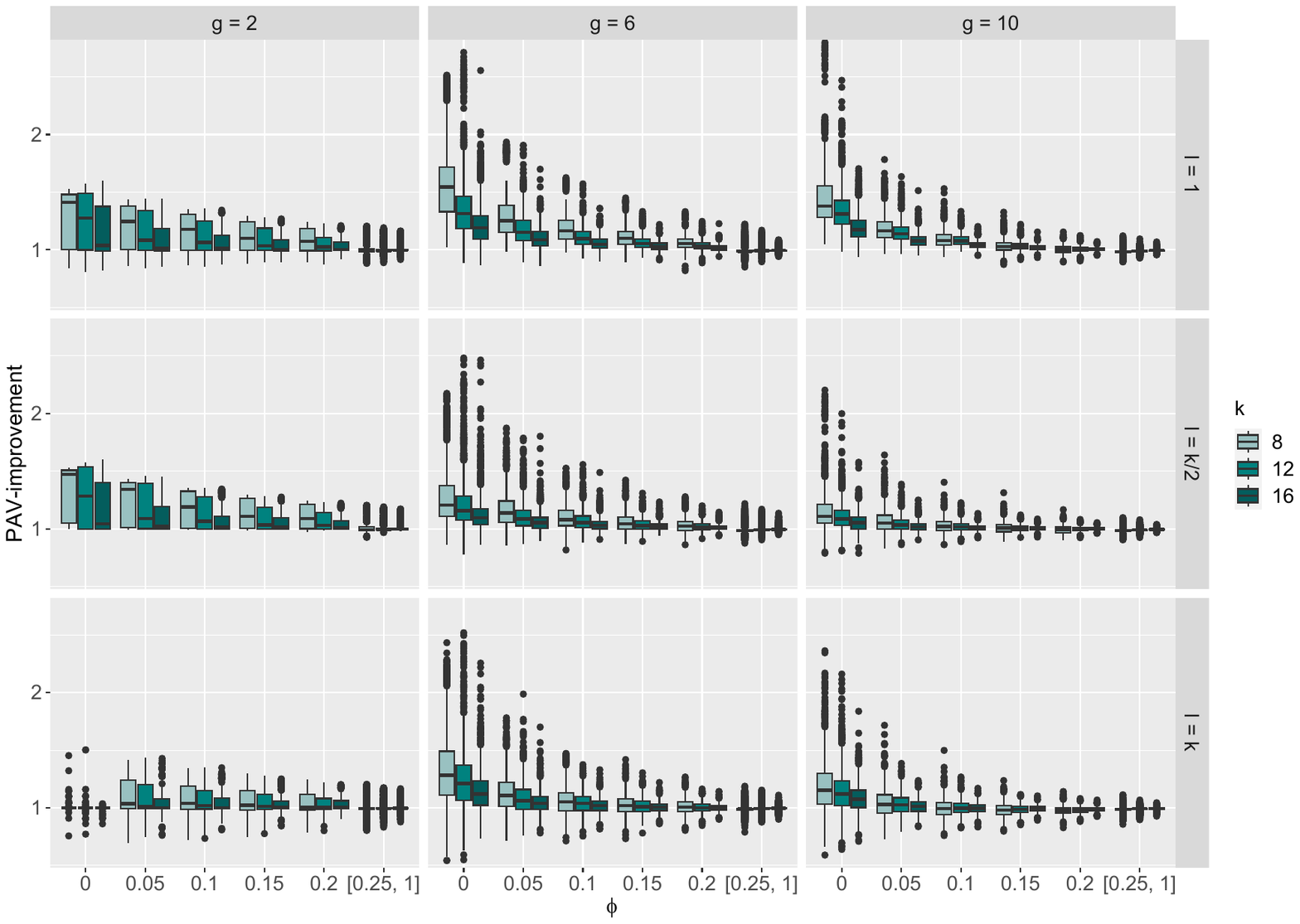}
    \caption{Boxplots of the PAV-
    improvement for different values of $\phi$, $k$, $g$, and $l$.
    }
    \label{fig:PAV-increase}
\end{figure*}


\subsection{Effect of variables other than $\phi$ on CC- and PAV-improvement}\label{app:variable_analysis}
To see what is the effect of the other variables on the relative increase in CC-score and PAV-score, we only consider the cases where $\phi=0$ (since 
in those we see the largest effect, and they are most interesting in practical considerations). 
Firstly, for $l=k$, we know from Theorem \ref{thm:CC-improvement_BP} and Proposition \ref{prop:PAV-improvement_BP} that the CC-/PAV-improvement in broadcasted party-list profiles is always 1, if the most popular party has at least $k$ candidates. We see that for $g=2$ both scores are indeed often 1, but for larger $g$ not anymore. This can be explained by the intuition that the larger $g$, the lower the probability that the most popular party has at least $k$ candidates.
Secondly, 
if $g=2$ the 
CC-improvement 
centers around 1 and around 2, see also the density plot in Figure \ref{fig:density_CC-improvement}
. 
\begin{figure}[b]
    \centering
    \includegraphics[width = \columnwidth]{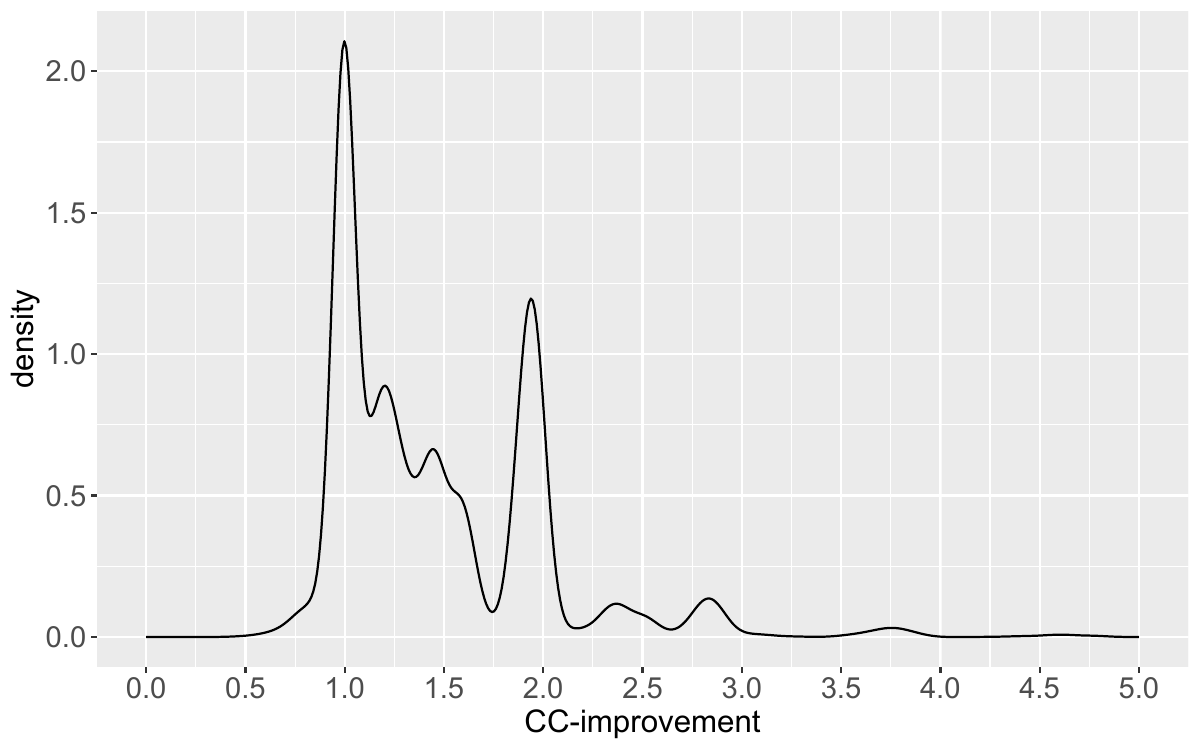}
    \caption{Probability density of CC-
    improvement for $g=2, \phi = 0$
    \vspace{10pt}
    }
    \label{fig:density_CC-improvement}
\end{figure}
This is because the two parties are likely to be similar in size, so if LV elects candidates from $P_2$ and AV doesn't, the CC-increase is around 2, otherwise it is 1.
For $g=6$ and $g=10$, the values of CC-improvement seem exponentially distributed, and indeed with a logarithmic transformation, visual inspection, or a quick qq-plot analysis, shows that they are roughly normally distributed. We perform an ANOVA on the logarithmic transformed data for $\phi = 0$ and $g=6$ or $g=10$ to study the influence of $k$, $l$, and $g$ on the CC-
improvement. It shows that all variables and their two-way interactions are relevant for the value of the CC-
improvement, therefore we study the effect sizes by the $\eta^2$ measure \citep{Adams2014}. 
Typically, partial $\eta^2$ above 0.01 is considered a small effect, above 0.06 is considered a medium effect and above 0.14 is considered a large effect.
\begin{table}[t]
     \caption{Effect sizes of 2-way ANOVA models of the logarithm of the CC-/PAV-improvements. 
     The first column shows the $\eta^2$ measure, the second the partial $\eta^2$ measure, which excludes variance from other independent variables. Values that indicate what is typically considered a large effect are highlighted.\\}
    \small
    \renewcommand{\arraystretch}{0.6}
    \begin{subtable}[h]{0.45\columnwidth}
        \caption{CC-improvement}
        \label{tab:effect_sizes_CC}
        \centering
        \begin{tabular}{rrr}
          \hline
         & $\eta^2$ & partial $\eta^2$ \\ 
          \hline
        g & 0.02 & 0.02 \\ 
          k & 0.12 & \textbf{0.15} \\ 
          l & \textbf{0.16} & \textbf{0.19} \\ 
          g:l & 0.00 & 0.00 \\ 
          g:k & 0.00 & 0.01 \\ 
          k:l & 0.02 & 0.03 \\ 
           \hline
        \end{tabular}
    \end{subtable}
    \hfill
    \begin{subtable}[h]{0.45\columnwidth}
    \caption{PAV-improvement}
    \label{tab:effect_sizes_PAV}
    \small
        \centering
        \begin{tabular}{rrr}
          \hline
         & $\eta^2$ & partial $\eta^2$ \\ 
          \hline
        g & 0.04 & 0.05 \\ 
          k & 0.07 & 0.09 \\ 
          l & 0.11 & \textbf{0.14} \\ 
          g:l & 0.07 & 0.09 \\ 
          g:k & 0.01 & 0.01 \\ 
          k:l & 0.01 & 0.01 \\ 
           \hline
        \end{tabular}
     \end{subtable}
\end{table}
In Table \ref{tab:effect_sizes_CC} we find that $k$ and $l$ explain large part of the variance (lower $k$ and $l$ cause higher CC-increase), which is what we would expect from the boxplots, but also $g$ (lower $g$ causes higher CC-increase) and the interaction between $k$ and $l$ (e.g. if $l=k$, the value of $k$ matters less than if $l=1$) explain a small part.

We do the same analysis for the PAV-
improvement. Again an ANOVA on the logarithmic transform of the PAV-
improvement for $\phi=0$ shows that all variables and their two-way interactions are significant. The $\eta^2$ measure gives the following (Table \ref{tab:effect_sizes_PAV}):
$l$ explains the largest part of the variance in the PAV-
improvement (with higher PAV-improvement for low $l$), then $k$ (also a higher score for low $k$) and the interaction between $g$ and $l$, and finally $g$ (
which mainly influences the variance).

We had expected that with smaller $l$ the welfare loss, measured as $\frac{s_{AV}(AV)}{s_{AV}(LV)}$, would increase, but an ANOVA test shows that there is only small effect of $l$ on the means of this value (partial $\eta^2 = 0.05$), and that for $g>2$, for both $l=1$ and $l=k$, this loss is higher than for $l=\frac{k}{2}$. A more fine-grained inspection of the `AV-improvement' (the inverse of the loss) with different values of $l$ shows that LV's welfare loss is lowest when $l$ is somewhere between 1 and $k$, see Figure \ref{fig:AV-loss}, except for the case where $g=2$, then $l=k$ gives no welfare loss. 
\begin{figure}[b]
    \centering
    \includegraphics[width=\columnwidth]{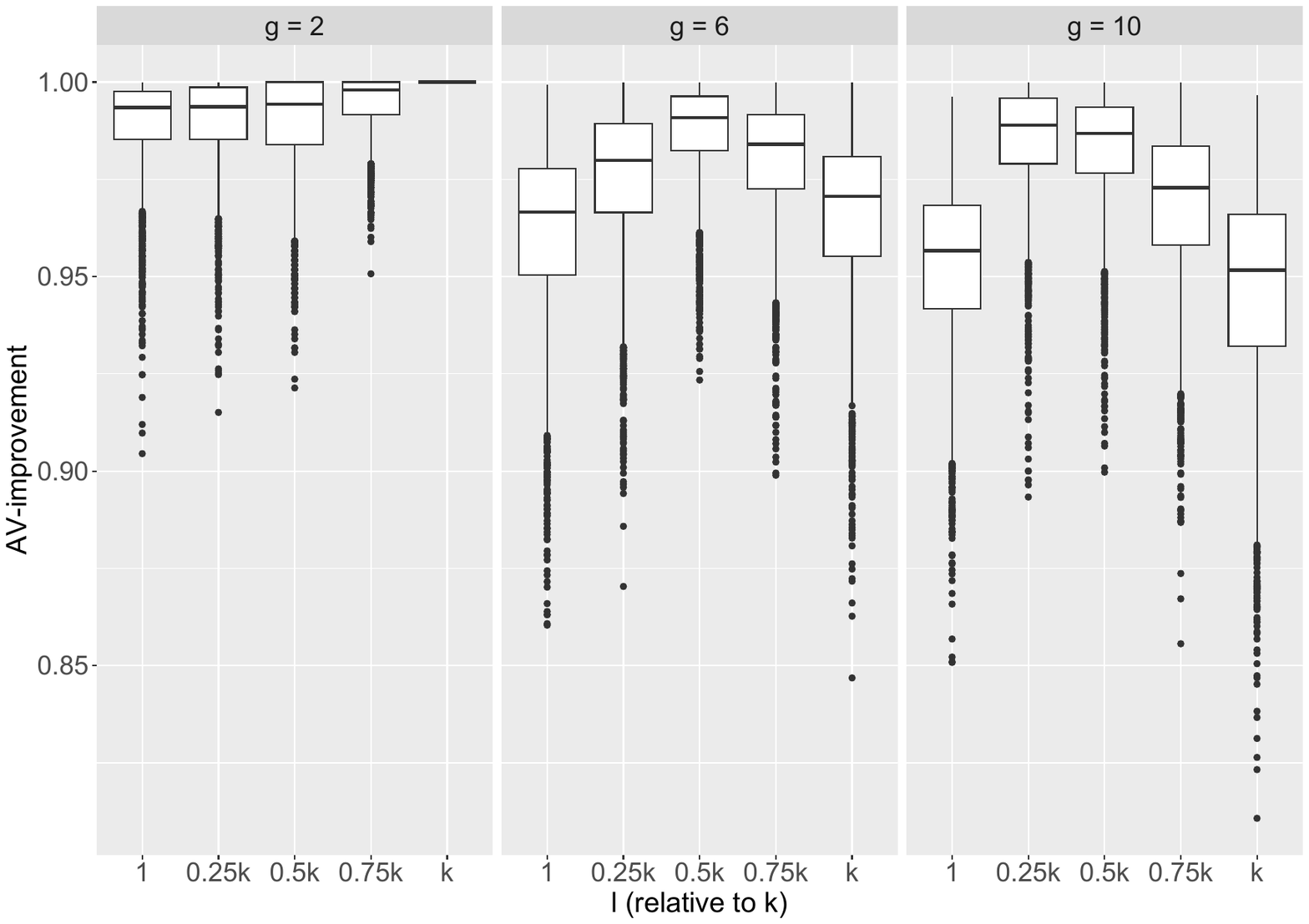}
    \caption{AV-improvement ($\frac{s_{AV}(LV)}{s_{AV}(AV)}$) for different values of $l$ (with $k\in \{8, 12, 16\}$, all values combined, $\phi=0$). Note that the `improvement' is always at most 1, because LV cannot have a higher AV-score than AV itself.}
    \label{fig:AV-loss}
\end{figure}

\subsection{Results for CC with voters partitioned randomly into parties.}
In Section \ref{sec:experiment} in the main paper, we used a model for generating profiles that made all parties have similar numbers of voters, by assigning a party uniformly at random to each voter. This may be considered a restrictive assumption on an election. We therefore ran the experiments with a slightly different model, that does not rely on this assumption. The model is almost the same to the disjoint model by \citet{szufa_2022} as explained above, with the difference that, instead of assigning a party randomly to each voter, it chooses a partition of the voters over the parties uniformly at random.

Since the differences between the sizes of the parties are now larger, we expect by Theorem \ref{thm:CC-improvement_BP} the CC-improvement for broadcasted party-list profiles to be smaller than in the case where the parties had more similar sizes. Figure \ref{fig:CC-improvement_partition} shows that this is indeed the case. It looks similar to the results of the simulation where parties had similar sizes (Figure \ref{fig:CC-improvement_all_phi}), but in general the CC-improvement is smaller, as one would expect with parties of more dissimilar sizes.
Table \ref{tab:effect_sizes_CC_random_partition} shows the $\eta^2$ effect sizes of an ANOVA model of the logarithm of the CC-improvement. 

\begin{table}[t]
     \caption{Effect sizes of 2-way ANOVA models of the logarithm of the CC-/PAV-improvements in elections where voters are distributed over parties by a random partition. 
     The first column shows the $\eta^2$ measure, the second the partial $\eta^2$ measure, which excludes variance from other independent variables. Values that indicate what is typically considered a large effect are highlighted.\\}
    \small
    \renewcommand{\arraystretch}{0.6}
    \begin{subtable}[h]{0.45\columnwidth}
    \caption{CC-improvement}
    \label{tab:effect_sizes_CC_random_partition}
        \centering
        \begin{tabular}{rrr}
          \hline
         & $\eta^2$ & partial $\eta^2$ \\ 
          \hline
        g & 0.03 & 0.04 \\ 
          k & 0.05 & 0.07 \\ 
          l & \textbf{0.16} & \textbf{0.18} \\ 
          g:l & 0.01 & 0.01 \\ 
          g:k & 0.00 & 0.00 \\ 
          k:l & 0.03 & 0.04 \\ 
           \hline
        \end{tabular}
    \end{subtable}
    \hfill
    \begin{subtable}[h]{0.45\columnwidth}
        \caption{PAV-improvement}
        \label{tab:effect_sizes_PAV_random_partition}
        \small
        \centering
        \begin{tabular}{rrr}
          \hline
         & $\eta^2$ & partial $\eta^2$ \\ 
          \hline
        g & 0.05 & 0.06 \\ 
          k & 0.00 & 0.00 \\ 
          l & \textbf{0.14} & \textbf{0.16} \\ 
          g:l & 0.02 & 0.02 \\ 
          g:k & 0.00 & 0.00 \\ 
          k:l & 0.01 & 0.01 \\ 
           \hline
        \end{tabular}
     \end{subtable}
\end{table}

We see in Table \ref{tab:effect_sizes_CC_random_partition} that $l$ is explaining most of the variance in the CC-improvement. Parameters $k$, $g$, and the interaction between $k$ and $l$ also explain part of the variance. These results are similar to the ones for the main experiment where parties had more similar sizes (Table \ref{tab:effect_sizes_CC}).

For the PAV-improvement, we expect different outcomes. From Proposition \ref{prop:PAV-improvement_BP} we know that if the relative difference between the sizes of the parties is large, LV will perform worse than AV on proportionality: the PAV-improvement will be smaller than 1. In Figure \ref{fig:PAV-improvement_partition}, we see that this is indeed the case: the improvement is often below 1. 
Table \ref{tab:effect_sizes_PAV_random_partition} shows the effect sizes of an ANOVA model of the logarithm of the PAV-improvement. 
We can read from Table \ref{tab:effect_sizes_PAV_random_partition} that $l$ is explaining almost all variance, but $g$ and the interaction between $g$ and $l$ are also important. Note that this differs from the results of the simulations with similar party sizes in the fact that $k$ does not explain the variance in PAV-improvement here.

With respect to the effect of the value of $l$ on the welfare loss, or `AV-improvement', this effect is similar to that in the previous setting (where voters are not distributed in a random partition over parties but rather choose a party uniformly at random), but stronger. See Figure \ref{fig:AV-loss_RP} (compare to Figure \ref{fig:AV-loss}).

\begin{figure}[t]
    \centering
    \includegraphics[width=\columnwidth]{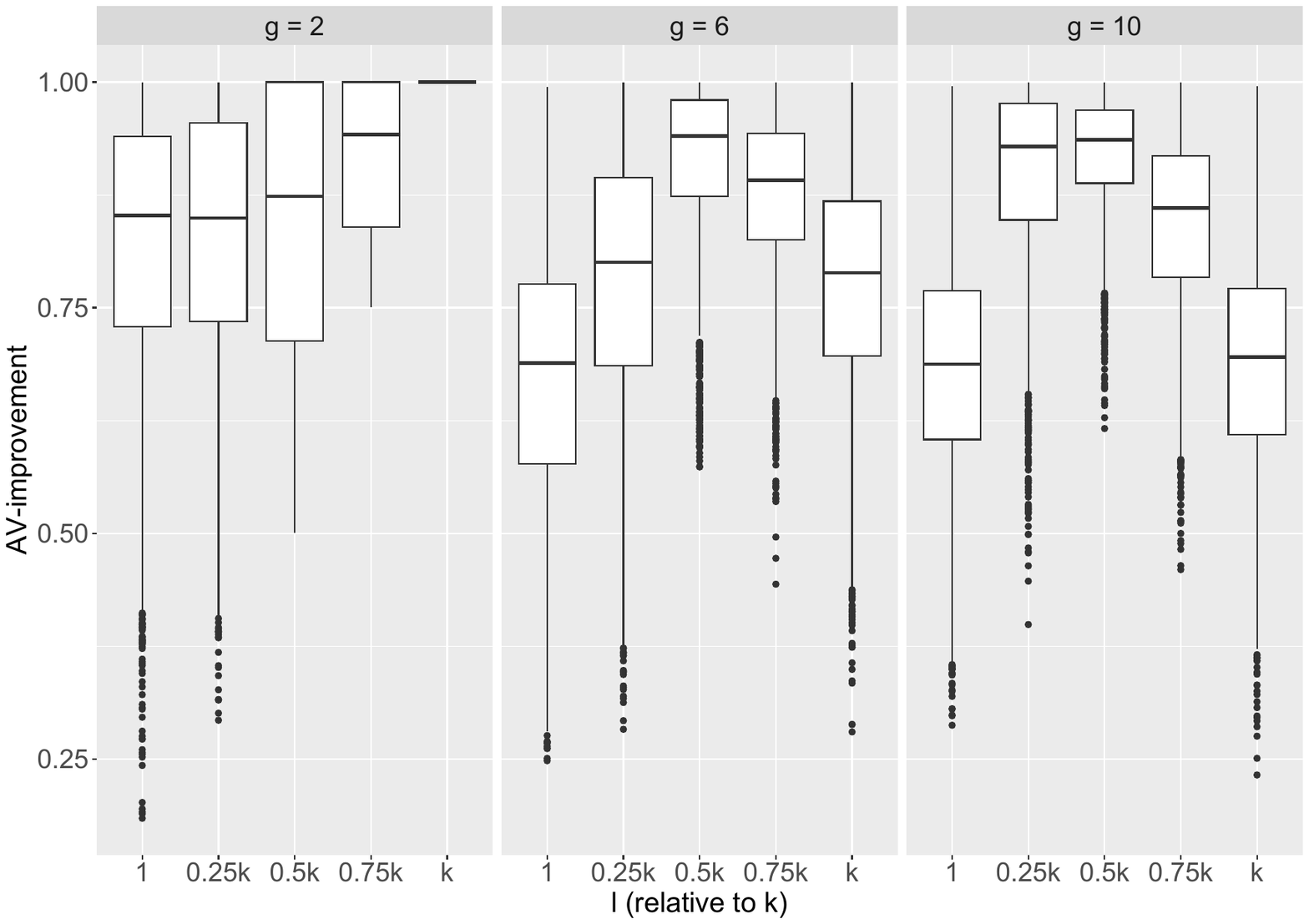}
    \caption{AV-improvement ($\frac{s_{AV}(LV)}{s_{AV}(AV)}$) for different values of $l$ (with $k\in \{8, 12, 16\}$, all values combined, $\phi=0$).  Voters are distributed over parties by a random partition. Note that the `improvement' is always at most 1, because LV cannot have a higher AV-score than AV itself.
    \vspace{15pt}
    }
    \label{fig:AV-loss_RP}
\end{figure}


\section{Utility vs. diversity}

\begin{figure}[t]
    \centering
    \includegraphics[width = 0.5\columnwidth]{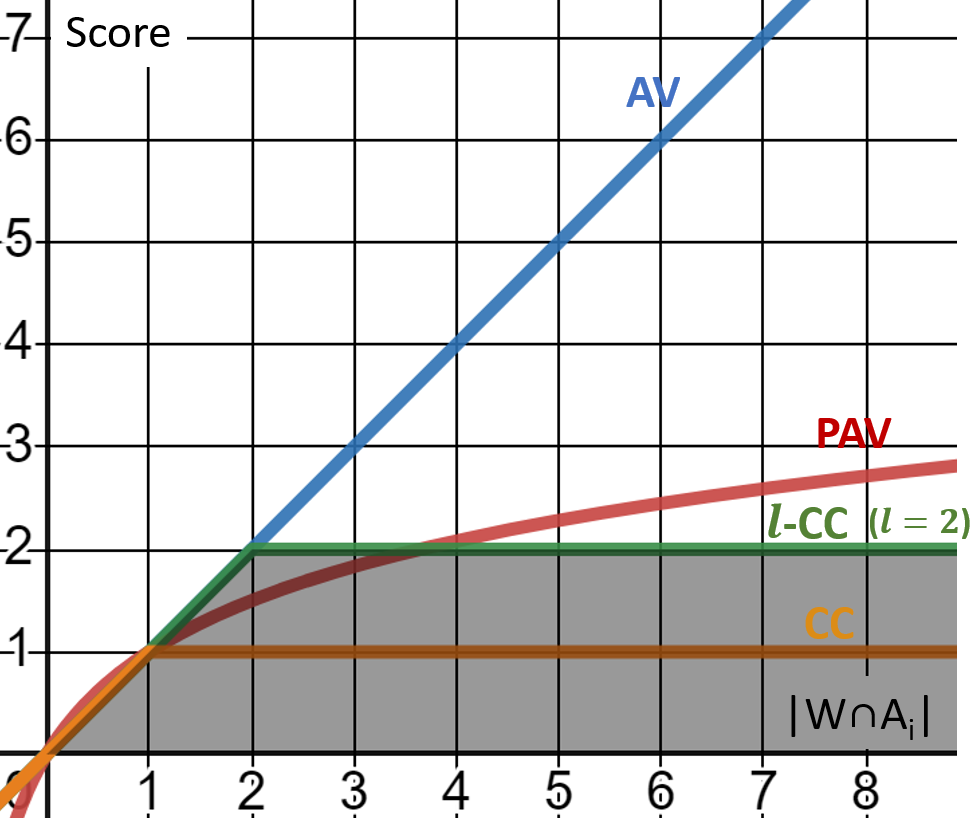}
    \caption{Scoring functions of approval based scoring rules
    \vspace{15pt}
    }
    \label{fig:scores}
\end{figure}

We can compare LV to approval based scoring rules and the utility and diversity that they provide. We can plot the score of a voter as a function of the number of approved candidates in the committee, like in \citet[Figure 2.2]{Lackner2023} or \citet[Figure 1]{Lackner2021}, see Figure \ref{fig:scores}. 
On the x-axis is the number of approved projects of a voter that is in the committee, on the y-axis the score that this number of approved projects provides to that voter. The rules AV, PAV, CC, and $l$-CC all choose the committee that gives the highest sum of scores over all voters. In general, rules that are higher in the graph provide more utility while rules that are lower provide more diversity. Since LV is not a scoring rule in this sense, we can not directly compare it to the given rules. However, we can think of LV choosing the committee that maximises the sum over voters of their score in terms of number of projects that a voter has in its ballot and that is in the committee $W$. Then we can draw this in the plot by considering the number of approved candidates of a voter that count towards this score. Maximally, this is $\min(l,|W\cup A_i|)$, the number of candidates in the voter's ballot that is elected, which is the same as the score for $l$-CC (the green line in the plot). However, it can also happen that approved candidates of a voter are only partially or not at all in her ballot, so they don't count for LV when it maximises the score. Therefore, the score of LV can be any value between 0 and the score of $l$-CC (the shaded area in the graph). Dependent on the value of $l$, this can give higher utility and lower diversity than PAV, but for any value of $l$ it can give lower utility and higher diversity than PAV (since always part of the shaded area is below the function for PAV).

\begin{figure*}[t]
    \centering
    \includegraphics[width=0.7\textwidth]{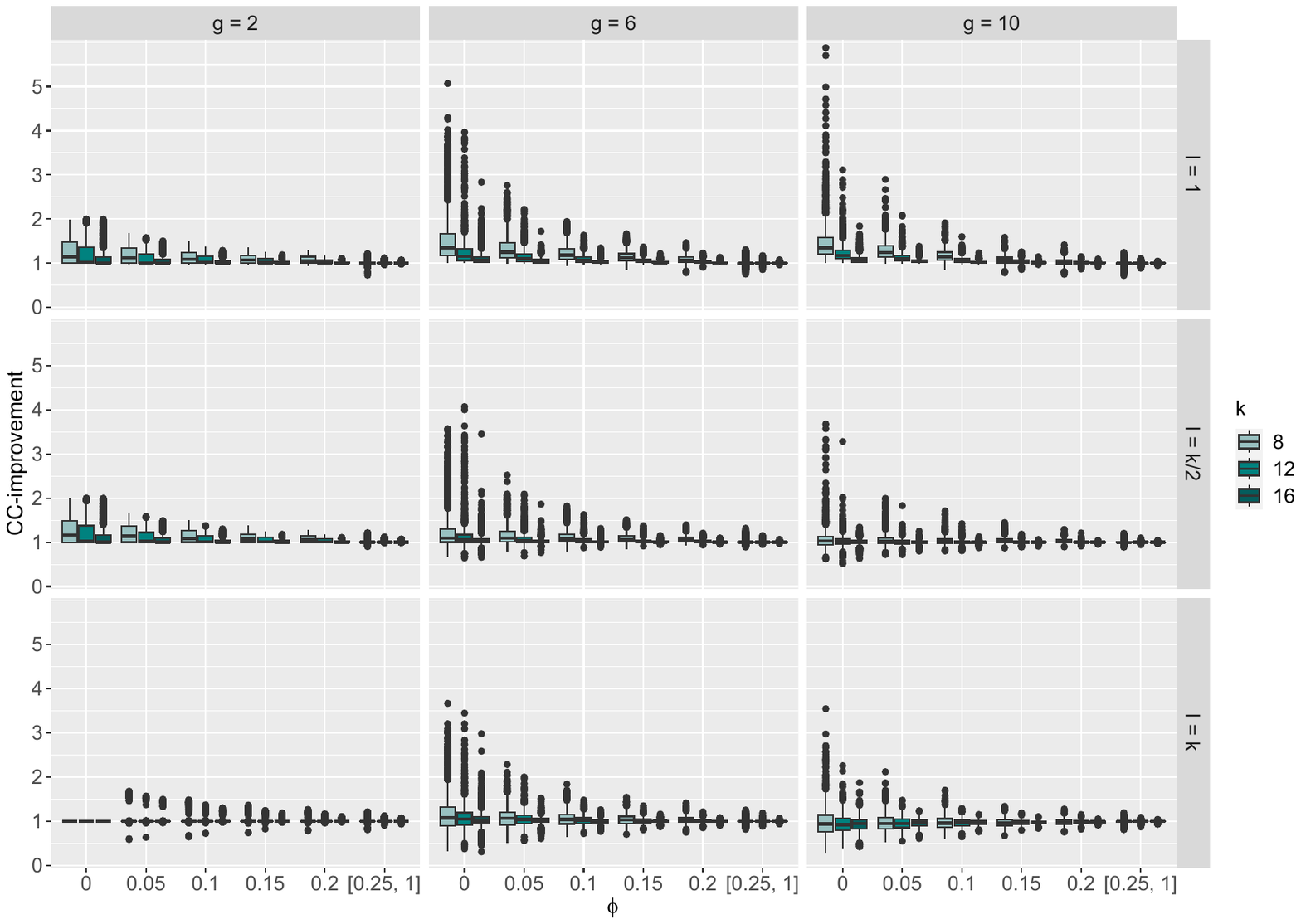}
    \caption{Boxplots of the CC-improvement for different values of $\phi$, $k$, $g$, and $l$. Voters are distributed over parties by a random partition.
    }
    \label{fig:CC-improvement_partition}
\end{figure*}

\begin{figure*}[t]
    \centering
    \includegraphics[width=0.7\textwidth]{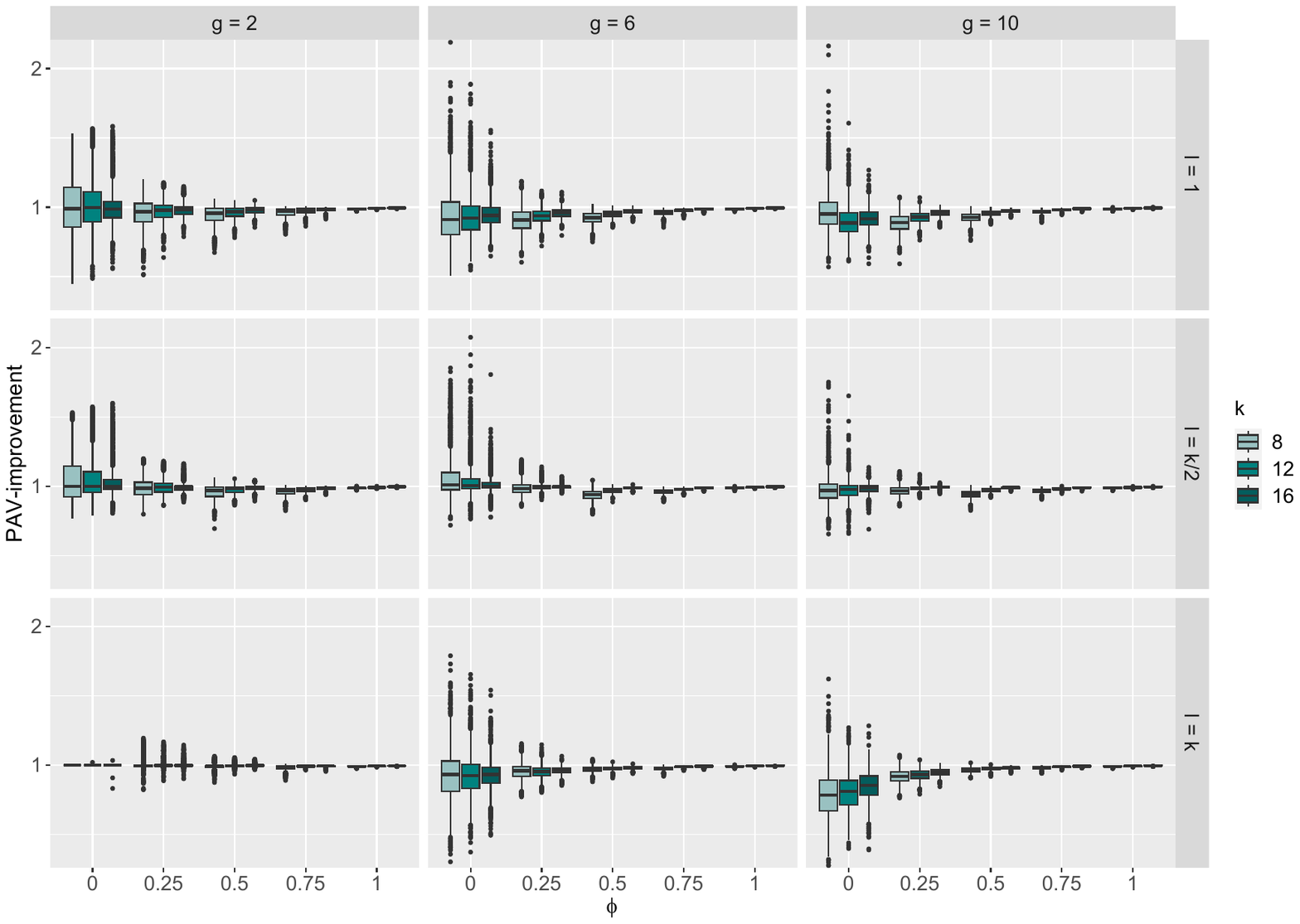}
    \caption{Boxplots of the PAV-improvement for different values of $\phi$, $k$, $g$, and $l$. Voters are distributed over parties by a random partition.
    }
    \label{fig:PAV-improvement_partition}
    \vspace{1cm}
\end{figure*}



\newpage

\begin{ack}
The authors wish to thank the anonymous reviewers of COMSOC'23, AAAI'24 and ECAI'24 for several very helpful suggestions on earlier versions of this work. The authors also wish to acknowledge Sophie van Schaik for initial contributions to the code for the experiments of Section \ref{sec:experiment}.
Zo\'{e} Christoff acknowledges support from the project Social Networks and Democracy (VENI project number Vl.Veni.201F.032) financed by the Netherlands
Organisation for Scientific Research (NWO). 
Davide Grossi acknowledges support by the 
Hybrid Intelligence Center, a 10-year program funded by the Dutch
Ministry of Education, Culture and Science through the Netherlands
Organisation for Scientific Research (NWO).
\end{ack}

\end{document}